\documentclass[submission,copyright,creativecommons]{eptcs}
\providecommand{\event}{TARK 2025} 

\usepackage{iftex}

\usepackage{url}
\usepackage{amsfonts, amsmath, amssymb, theorem}
\usepackage{graphics}
\usepackage{thmtools,thm-restate}
\usepackage{tikz}
\usetikzlibrary{fit,shapes.geometric}

\usepackage{newproof}
\theorembodyfont{\rmfamily}

\newtheorem{theorem}{Theorem}

\newtheorem{proposition}[theorem]{Proposition}

\newcommand{\eq}{\leftrightarrow}

\newcommand{\imp}{\rightarrow}

\newcommand{\et}{\wedge}
\newcommand{\vel}{\vee}
\newcommand{\Et}{\bigwedge}
\newcommand{\Vel}{\bigvee}

\newcommand{\all}{\forall}
\newcommand{\is}{\exists}

\newcommand{\Dia}{\Diamond}
\newcommand{\dia}[1]{\langle #1 \rangle}

\renewcommand{\phi}{\varphi}

\newcommand{\weg}[1]{}

\newcommand{\Formulas}{{\mathcal L}}

\newcommand{\langu}{\Formulas}

\newcommand{\Domain}{{\mathcal D}}
\newcommand{\domain}{\Domain}

\newcommand{\lang}{\langu}

\newcommand{\knows}{\Box}

\usepackage{tikz}

\newcommand{\red}[1]{{\color{red}#1}}

\newcommand{\bisimilar}{\leftrightarrows}
\newcommand{\refines}{\leftleftarrows}
\newcommand{\simulates}{\rightrightarrows}

\renewcommand{\implies}{\rightarrow}
\renewcommand{\iff}{\leftrightarrow}

\newcommand{\suspects}{\Dia}
\newcommand{\cover}{\nabla}

\newcommand{\rmlbox}{[\text{\footnotesize{$\leftleftarrows$}}]}
\newcommand{\rmldia}{\dia{\text{\footnotesize{$\leftleftarrows$}}}}
\newcommand{\smldia}{\dia{\text{\footnotesize{$\rightrightarrows$}}}}
\newcommand{\smlbox}{[\text{\footnotesize{$\rightrightarrows$}}]}
\newcommand{\originbox}{[\circ]}

\newcommand{\agents}{A}

\newcommand{\ol}[1]{\ensuremath{\overline{#1}}}

\newcommand{\suggestion}[1]{{\color{blue}#1}}

\ifpdf
  \usepackage{underscore}         
  \usepackage[T1]{fontenc}        
\else
  \usepackage{breakurl}           
\fi

\title{Modal Logic for Simulation, Refinement, and Mutual Ignorance}
\author{Hans van Ditmarsch
\institute{University of Toulouse\\ CNRS, IRIT, France}
\email{hansvanditmarsch@gmail.com}
\and
Tim French
\institute{University of Western Australia\\ Perth, Australia}
\email{tim.french@uwa.edu.au}
\and
Rustam Galimullin
\institute{University of Bergen\\ Norway}
\email{rustam.galimullin@uib.no}
\and
Louwe B. Kuijer 
\institute{Liverpool University\\ UK}
\email{lbkuijer@liverpool.ac.uk}
}
\def\titlerunning{Simulation, Refinement, and Mutual Ignorance}
\def\authorrunning{H. van Ditmarsch, T. French, R. Galimullin \& L.B. Kuijer}
\begin{document}
\maketitle

\begin{abstract}
Simulation and refinement are variations of the bisimulation relation, where in the former we keep only \textbf{atoms} and \textbf{forth}, and in the latter only \textbf{atoms} and \textbf{back}. Quantifying over simulations and refinements captures the effects of information change in a multi-agent system. In the case of quantification over refinements, we are looking at all the ways the agents in a system can become more informed. Similarly, in the case of quantification over simulations, we are dealing with all the ways the agents can become less informed, or in other words, could have been less informed, as we are at liberty how to interpret time in dynamic epistemic logic. While quantification over refinements has been well explored in the literature, quantification over simulations has received considerably less attention. 
In this paper, we explore the relationship between refinements and simulations. To this end, we also employ the notion of mutual factual ignorance that allows us to capture the state of a model before agents have learnt any factual information. In particular, we consider the extensions of multi-modal logic with the simulation and refinement modalities, as well as modalities for mutual factual ignorance. We provide reduction-based axiomatizations for several of the resulting logics that are built extending one another in a modular fashion.
\end{abstract}

\section{Introduction}  \label{sec.intro}

To a multi-agent modal logic, and in particular to 
epistemic logic \cite{hintikka:1962}, we can add modalities that are quantifiers over arbitrary information change, called `refinement quantifiers'. In {\em refinement modal logic} (RML) \cite{hvdetal.loft:2009,hvdetal.felax:2010,halesetal:2011,hales:2011,bozzellietal.inf:2014,hales:2016, bozzellietal.tcs:2014,achilleosetal:2013,xingetal:2019} one can intuitively say by $\rmldia\phi$ 
that there is some public or private change of information after which the formula $\phi$ is true. It is known that it is equivalent to saying that there is an action model \cite{baltagetal:1998} such that after its execution, $\phi$ is true; as in arbitrary action model logic \cite{hales2013arbitrary}. However, note that we then do not have modalities for action models in the logical language. Instead, in RML we quantify over the refinements of a given model. From the requirements {\bf atoms}, {\bf forth} and {\bf back} of a bisimulation \cite{blackburnetal:2001}, a refinement relation only needs to satisfy {\bf atoms} and {\bf back}. The semantics of the refinement quantifier directly uses this refinement relation. 

In this contribution we investigate the dual logic, called {\em simulation modal logic} (SML), that quantifies over the simulations of a given model. Now, from the requirements of a bisimulation, a simulation relation only needs to satisfy {\bf atoms} and {\bf forth}. Just as RML describes, in an epistemic setting, all the intricate ways in which agents can become more informed about the current state of information, in simulation modal logic we describe how agents can become less informed about the current state of information. Becoming less informed can be seen as a kind of forgetting \cite{suetal:2009,langetal:2010,hvdetal.introKRA:2009,DuqueNSSV15,FangLD19,liang24}, but possibly a more appealing intuition is to imagine a prior state of information to the current one, wherein agents were still not as informed and knowledgeable as they are now: just as refinement is a kind of \emph{belief expansion}, simulation is therefore a kind of \emph{belief contraction} \cite{agm:1985}.  Refinement modal logic and simulation modal logic are dynamic epistemic logics: logics with epistemic modalities as well as dynamic modalities interpreted as model updates \cite{hvdetal.del:2007,baltagetal.hpi:2008,moss.handbook:2015}. Although there are various ways to deal with belief contraction in a modal epistemic logic \cite{demolombeetal:2003,aucher.planning:2012,hvdetal.aaai:2007}, no general mechanism for belief contraction in dynamic epistemic logic has been proposed to our knowledge. Information updates such as public announcements amount to belief expansion \cite{hvdetal.aimlbook:2005,aucher.phd:2008}, whereas dynamic belief revision, often using Kripke models with more structure such as plausibility relations or preference relations have also found widespread use \cite{jfak.jancl:2007,baltagetal.tlg3:2008,hvd.prolegomena:2005,jfaketal.handbook:2015}. There are two different ways to represent belief contraction in dynamic epistemic logic: logical languages with dynamic modalities referring to what was true in the past given a history-based semantics \cite{Sack10,BalbianiDH16,baltagetal:2022}, and logical languages with dynamic modalities updating the current model with respect to some formula parametrizing it \cite{hvdetal.introKRA:2009,DuqueNSSV15}. Examples of the former are the dynamic epistemic logic with a `yesterday' modality \cite{Sack10}, the logic of what is true before a public announcement \cite{BalbianiDH16}, and arbitrary public announcement logic with memory \cite{baltagetal:2022}. Examples of the latter are \cite{hvdetal.introKRA:2009}, merely forgetting atomic propositions, and \cite{DuqueNSSV15}, forgetting arbitrary modal formulas. One could say that \cite{DuqueNSSV15} comes closest to the status quo of belief contraction in dynamic epistemic logic. With our proposed simulation modal logic we wish to further contribute to the so far quite restricted literature on belief contraction in dynamic epistemic logic in a meaningful way: it does not merely represent what was known before an announcement, a public action, but what was known before any informative action, such as also private announcements.

Despite simulation and refinement being dual, there is still a somewhat fundamental difference between the two, that can be summarized by saying that there are far more ways to become less informed than there are ways to become more informed. For example, given a publicly announced formula $\phi$ and an epistemic model $M$ encoding what agents know, there is a unique updated model resulting from that public announcement, but there are many epistemic models such that updating them with the announcement of $\phi$ results in the model $M$.

In this contribution we also axiomatize SML. In order to axiomatize SML we follow two different paths. First, we propose an axiomatization that can be seen as the dual of that of refinement modal logic, however, with the difference that we have consistency requirements on formulas occurring in one of the axioms (the axiom involves knowledge of a disjunction of simulations of formulas which can be consistent even if one of the disjuncts is inconsistent). Second, we propose an integrated axiomatization of refinement modal logic and simulation modal logic avoiding this complication of the consistency requirement. The need for refinement modalities can be explained from the semantics: we compare any given information state (pointed Kripke model) $M_s$ wherein we wish to evaluate a formula $\smldia\phi$ (there is simulation after which $\phi$ is true) to the (pointed) model $M^\circ_t$ of (maximal) mutual factual ignorance 
(also known as the epistemic model of {\em blissful ignorance} \cite{hvdetal.ckcb:2009}). We typically do not have that $M$ is a submodel of $M^\circ$, and also not that $\phi$ is true in $M^\circ$, but that $\phi$ is true in a pointed model $M'_{s'}$ such that $M'$ is a refinement of $M^\circ$. Model $M'$ is so to speak
`in between $M$ and $M^\circ$': we then have that $M \simulates M' \simulates M^\circ$ ($M^\circ$ simulates $M'$ and $M'$ simulates $M$, or dually, $M$ refines $M'$ and $M'$ refines $M^\circ$). This complication also requires us to refer to truth in the mutual factual ignorance model $M^\circ$ by a dedicated modal construct $\originbox\psi$, for `$\psi$ was \emph{originally} true'. The consistency of constituents of $\phi$ in the above $\smldia\phi$ then translates into the requirement that these constituents of $\phi$ can become true after a refinement of the mutual factual ignorance model, meaning that they are satisfiable and therefore consistent.

In this paper we examine the relationships between the three 
modalities, $\smlbox$, $\rmlbox$ and $\originbox$, and consider their logics.
We start by recalling the established results for RML and SML, 
and present novel axiomatizations for each.
Then we consider the combination of simulation and refinement, and in so doing also introduce {\em origin modal logic} OML. 
We axiomatize OML on its own, and furthermore in combination with simulation modalities and refinement quantifiers, in the axiomatization {\bf ROSML}. It should be noted that all our results are for epistemic models with arbitrary accessibility relations for each agent, and not with equivalence relations. Although fair game, the intended epistemic models in dynamic epistemic logic are with equivalence relations, in which case the epistemic modality represents knowledge. The extension of our results to knowledge is left for future research.

\section{Syntax and Semantics}


\paragraph{Languages}

Given a non-empty countable (finite or infinite) set of {\em propositional variables} ({\em atoms}) $P$ and a non-empty finite set of {\em agents} $A$, we consider a combined language with refinement and simulation quantifiers, as well as the origin modality. The elements of the language are the {\em formulas}. 

\[ \begin{array}{lclll} 
\lang & \ni & \phi & ::= & p \mid \neg \phi \mid (\phi \et \phi) \mid \Box_a \phi  \mid \rmlbox\phi \mid \smlbox\phi \mid \originbox \phi 
\end{array}
\] 
where $p \in P$, $a \in A$. The language $\lang^{\Box\refines}$ is the fragment with only $\Box_a$ and $\rmlbox$ modalities, the language $\lang^{\Box\simulates}$ is the fragment with only $\Box_a$ and $\smlbox$ modalities, the language $\lang^{\Box\circ}$ is the fragment with only $\Box_a$ and $\originbox$ modalities, and $\lang^\Box$ is the fragment with only the $\Box_a$ modalities. The propositional fragment is $\lang_0$. Other propositional connectives are defined by abbreviation, and we also let $\Dia_a \phi:=\neg\Box_a\neg\phi$. In an epistemic setting, for $\Dia_a \phi$ we read `agent $a$ considers $\phi$ possible', and for $\Box_a \phi$, `agent $a$ knows $\phi$'. 
For $\rmlbox\phi$ we read `after any refinement, $\phi$', and for $\smlbox\phi$, `after any simulation, $\phi$'. We also define $\rmldia\phi$ as $\neg\rmlbox\neg\phi$, and $\smldia\phi$ as $\neg\smlbox\neg\phi$. For $\originbox\phi$ we read `originally $\phi$ (was true)'. We will call $\originbox$ the \textit{origin modality}. Finally, the \textit{cover} operator $\cover_a\Phi$, where $\Phi$ is a finite set of formulas, is $\Et_{\phi\in\Phi} \Dia_a \phi \et \Box_a \Vel_{\phi\in\Phi} \phi$.

Given a formula $\phi$, 
$d(\phi)$ is the {\em modal depth} or $\Box$-depth of $\phi$. 
The definition of $\Box$-depth is: $d(p) = 0$, $d(\neg\phi) = d(\rmlbox\phi) = d(\smlbox\phi)= d(\originbox \phi)= d(\phi)$, $d(\phi\et\psi) = \max \{d(\phi),d(\psi)\}$, $d(\Box_a\phi) = d(\phi)+1$.

\paragraph{Structures}

An {\em epistemic model} (or `Kripke model', or just `model') $M = (S, R, V)$ consists of a non-empty countable {\em domain} $S$ of {\em states} (or `worlds'; the domain is also denoted $\domain(M)$), an {\em accessibility function} $R: A \imp {\mathcal P}(S \times S)$, where each $R(a)$ is an {\em accessibility relation}, and a {\em valuation} $V: S \imp {\mathcal P}(P)$; each $V(s)$ is the set of atoms that are true in $s$. 
For $s \in S$, a pair $(M,s)$, for which we write $M_s$, is a {\em pointed (epistemic) model}. 
The model class without any restrictions is K. The class of models where all accessibility relations are equivalence relations is S5. 

The \emph{mutual factual ignorance model} $M^\circ$ is defined as $M = (S^\circ, R^\circ, V^\circ)$, where $S^\circ$ is the set of all subsets of $P$, $R^\circ_a = S^\circ \times S^\circ$, and  $V(s)=s$ for each $s \in S^\circ$. 
Given that $P$ is countably infinite, 
$|S^\circ|$ is not countable. 


Let models $M= (S, R, V)$ and $M'= (S', R', V')$ be given. A non-empty relation $Z \subseteq S \times S'$ is a \emph{bisimulation} between $M$ and $M'$, notation $Z: M \bisimilar M'$, if for all pairs $(s,s') \in
Z$ 
and $a\in A$:
\begin{description}
\item[atoms] $V(s)=V'(s')$;
\item[forth] if
$R_ast$, then there is a $t'\in
S'$ such that $R'_as't'$ and
$Ztt'$;
\item[back] if $R'_as't'$, then there is a
$t \in S$ such that $R_ast$
and $Ztt'$.
\end{description}
We write $M \bisimilar M'$ if there is a bisimulation between $M$ and $M'$, and we write $M_s \bisimilar M'_{s'}$ if there is a bisimulation between $M$ and $M'$ containing pair $(s,s')$. 

Similarly, a {\em simulation} satisfies {\bf atoms} and {\bf forth}, where we write $M \simulates M'$ if there is a simulation between $M$ and $M'$, and $M_s \simulates M'_{s'}$ if it contains pair $(s,s')$; whereas its dual, the {\em refinement}, satisfies {\bf atoms} and {\bf back}, and where we write $M \refines M'$ if there exists a refinement between $M$ and $M'$, and $M_s \refines M'_{s'}$ if it contains pair $(s,s')$. 
Further, if $M \refines M'$, the model $M'$ is  called a refinement of $M$ as well. It is sometimes confusing that in the literature the term `refinement' denotes both the refinement relation linking $M$ and $M'$, and the refined model $M'$.



\paragraph{Semantics}

Assume an epistemic model $M = (S, R, V)$, and let $s \in S$. We define $M_s \models \phi$ (for: $M_s$ \emph{satisfies} $\phi$, or $\phi$ is \emph{true} in $M_s$) by induction.
$$\begin{array}{lcl}
M_s \models p &\mbox{ \ iff \ } & p \in V(s) \\ 
M_s \models \neg \phi &\mbox{iff} & M_s \not \models \phi \\ 
M_s \models \phi \et \psi &\mbox{iff} & M_s \models \phi  \text{ and } M_s \models \psi \\  
M_s \models \Box_a \phi &\mbox{iff} & M_t  \models \phi \text{ for all }  t \in S \text{ such that } (s,t) \in R_a \\
M_s \models \rmlbox \phi & \mbox{iff} & M'_{s'} \models \phi \text{ for all } M'_{s'} \text{ such that } M_s \refines M'_{s'} \\
M_s \models \smlbox \phi & \mbox{iff} & M'_{s'} \models \phi \text{ for all } M'_{s'} \text{ such that } M_s \simulates M'_{s'} \\
M_s \models \originbox \phi & \mbox{iff} & M^\circ_{V(s)} \models \phi
\end{array} $$

Let us give an example. 
A refinement of a given model need not be a submodel and a simulation of a given model need not contain the given model.  Consider the models depicted below. Assume that $i$ and $-i$ states have the same valuation of atoms.
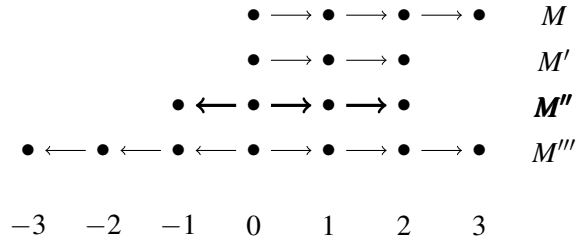
\begin{center}
\begin{tikzpicture}[->]
\node (m) at (5,0) {$M$};
\node (10) at (1,0) {$\bullet$};
\node (20) at (2,0) {$\bullet$};
\node (30) at (3,0) {$\bullet$};
\node (40) at (4,0) {$\bullet$};
\draw (10) -> (20);
\draw (20) -> (30);
\draw (30) -> (40);
\node (mp) at (5,-.6) {$M'$};
\node (10b) at (1,-.6) {$\bullet$};
\node (20b) at (2,-.6) {$\bullet$};
\node (30b) at (3,-.6) {$\bullet$};
\draw (10b) -> (20b);
\draw (20b) -> (30b);
\node (mpp) at (5,-1.2) {\pmb{$M''$}};
\node (00bb) at (0,-1.2) {$\bullet$};
\node (10bb) at (1,-1.2) {$\bullet$};
\node (20bb) at (2,-1.2) {$\bullet$};
\node (30bb) at (3,-1.2) {$\bullet$};
\draw[very thick] (10bb) -> (00bb);
\draw[very thick] (10bb) -> (20bb);
\draw[very thick] (20bb) -> (30bb);
\node (mppp) at (5,-1.8) {$M'''$};
\node (10bbb) at (1,-1.8) {$\bullet$};
\node (20bbb) at (2,-1.8) {$\bullet$};
\node (30bbb) at (3,-1.8) {$\bullet$};
\node (40bbb) at (4,-1.8) {$\bullet$};
\node (20n) at (0,-1.8) {$\bullet$};
\node (30n) at (-1,-1.8) {$\bullet$};
\node (40n) at (-2,-1.8) {$\bullet$};
\draw (10bbb) -> (20bbb);
\draw (20bbb) -> (30bbb);
\draw (30bbb) -> (40bbb);
\draw (10bbb) -> (20n);
\draw (20n) -> (30n);
\draw (30n) -> (40n);
\node (00x) at (1,-2.8) {$0$};
\node (10x) at (2,-2.8) {$1$};
\node (20x) at (3,-2.8) {$2$};
\node (30x) at (4,-2.8) {$3$};
\node (10nx) at (0,-2.8) {$-1$};
\node (20nx) at (-1,-2.8) {$-2$};
\node (30nx) at (-2,-2.8) {$-3$};

\end{tikzpicture}
\end{center}
We note that $M'$ is a submodel (model restriction) of $M$ and $M''$ is a refinement of $M$ that is not a submodel. However, $M''$ is a submodel of $M'''$ and $M'''$ is a bisimilar copy of $M$. Simulations are duals of refinement. So, $M$ is a simulation of $M''$ even though $M''$ is not a submodel of $M$.

In an epistemic setting, one could imagine a two-agent situation where $a$ and $b$ both know that a propositional variable $p$ is true, but are both uncertain whether the other one knows, as encoded in the model $M$ below in state $s$. A simulation of this is the two-state mutual factual ignorance model $M^\circ$. It is clear that $M^\circ$ does not contain $M$. But $M$ is a refinement of bisimilar copy $M''$ of $M^\circ$. The model $M^\circ$ can be seen as a previous state of information that $a$ and $b$ were in, where they were both informed that $p$ is true while they remained both uncertain whether the other was also informed.

\medskip

\begin{tikzpicture}
\node (m) at (-1,0) {$M:$};
\node (0) at (0,0) {$\neg p$};
\node (1) at (2,0) {$p$};
\node (2) at (4,0) {$p$};
\node (2a) at (4,0.3) {$s$};
\node (3) at (6,0) {$p$};
\node (4) at (8,0) {$\neg p$};
\draw[-] (0) -- node[above] {$a$} (1);
\draw[-] (1) -- node[above] {$b$} (2);
\draw[-] (2) -- node[above] {$a$} (3);
\draw[-] (3) -- node[above] {$b$} (4);
\end{tikzpicture}

\begin{tikzpicture}
\node (m) at (-1,0) {$M'':$};
\node (0) at (0,0) {$\neg p$};
\node (1) at (2,0) {$p$};
\node (2) at (4,0) {$p$};
\node (2a) at (4,0.3) {$s$};
\node (3) at (6,0) {$p$};
\node (4) at (8,0) {$\neg p$};
\draw[-] (0) -- node[above] {$ab$} (1);
\draw[-] (1) -- node[above] {$ab$} (2);
\draw[-] (2) -- node[above] {$ab$} (3);
\draw[-] (3) -- node[above] {$ab$} (4);
\end{tikzpicture}

\begin{tikzpicture}
\node (m) at (-5,0) {$M^\circ:$};
\node (0) at (0,0) {$p$};
\node (0a) at (0,0.3) {$s$};
\node (1) at (2,0) {$\neg p$};
\draw[-] (0) -- node[above] {$ab$} (1);
\end{tikzpicture}



\section{Simulation, Refinement, and Mutual Factual Ignorance}\label{sec:examples}

\paragraph{Refinement}
We review some results obtained for refinement modal logic over the years for which we present a novel axiomatization in the next section.

In \cite{bozzellietal.inf:2014} and most other works on RML, instead of the refinement relation as we defined above (for the set of all agents), the refinement relation has a parameter $G \subseteq A$, where apart from {\bf atoms} and {\bf back} for all agents, additionally, {\bf forth} is required for the agents in $A \setminus G$. Here we only consider refinement for the set $A$ of all agents, so that the only requirement is {\bf atoms} and {\bf back} for all agents. This simplifies the presentation. 

As the refinement relation is transitive, reflexive, and confluent (for all $x,y,z$, if $x \refines y$ and $x \refines z$, there is $w$ such that $y \refines w$ and $z \refines w$), and we additionally have atomicity (there is a maximal refinement: remove all relations!), we  have for the refinement modality the corresponding validities:
\begin{itemize}
\item $\rmlbox \phi \imp \phi$ \quad ({\bf T}$^\refines$)
\item $\rmldia \rmldia \phi \imp \rmldia \phi$ \quad ({\bf 4}$^\refines$) 
\item $\rmldia \rmlbox \phi \imp \rmlbox \rmldia \phi$ \quad (Church-Rosser) / {\bf CR}$^\refines$
\item $\rmlbox \rmldia \phi \imp \rmldia \rmlbox \phi$ \quad (McKinsey) / {\bf MK}$^\refines$ 
\end{itemize} 

The requirement of {\bf atoms} and {\bf back} entails that a refinement of a given model is a restriction of a bisimilar copy of that model. Analogously, in the logical language, refinement quantification is bisimulation quantification followed by relativization. In {\em bisimulation quantified} logics we have explicit quantifiers $\all p$ over propositional variables $p$ \cite{french:2006,visser:1996}. {\em Relativization} $\phi^p$ of a formula $\phi$ to an atomic proposition $p$ is a syntactic way to describe model restrictions, such as the consequences of a public announcement: $(M|p)_s \models \phi$ iff $M_s \models \phi^p$, which in public announcement logic is the same as $M_s \models \dia{p}\phi$, see \cite{milleretal:2005,jfak.book:2011} (here $M|p$ is the restriction of the model to the state where $p$ is true, and $\dia{p}\phi$ means `$p$ is true and after its announcement $\phi$ is true'). In RML, given $\phi \in \lang^\Box$, we have \cite[Sect.\ 4.3]{bozzellietal.inf:2014}: \[ \rmldia\phi \text{ is equivalent to } \is p \dia{p}\phi \] 
Furthermore, the refinement quantifier can be seen as a quantification over {\em action models} \cite{baltagetal:1998}. Already on the level of the semantics we have that:
\begin{quote} {\em Executing an action model produces a refinement of the initial model, and for every refinement of a finite model there is an action model producing it.} \end{quote}
On the level of the language we can, given $\rmldia\phi$, synthesize an action model $U$ with the same update effect (such that $\dia{U}\phi$ is true). For details, see \cite{hales2013arbitrary} and the related (for a different but comparable update mechanism, called `arrow update') \cite{hvdetal.aus:2020}.

The axiomatization of RML reported in \cite{bozzellietal.inf:2014} is in terms of parametrized refinement quantifiers $\dia{\refines_a}$ for $a \in A$ (interpreting $a$-refinement relations $\refines_a$) with which one can also define $\dia{\refines_G}$ for $G \subseteq A$. We recall that we now only consider $\refines$, that is, $\refines_A$). We therefore wish to report in this contribution an alternative axiomatization {\bf RML} for RML with only refinement quantifiers $\refines$ (Section~\ref{sec:axiomatization}). It resembles the quantifier part of the axiomatization of arbitrary arrow update model logic \cite{hvdetal.aus:2020}, and was reported in \cite{Ditmarsch23} without proof of soundness and completeness. 
It is simpler than the axiomatization in \cite{bozzellietal.inf:2014} and is also based on reduction axioms. 

\paragraph{Simulation}

Perhaps the most relevant work here is \cite{xingetal:2019}, where the authors explore the logic that  combines refinement and simulation in a single modality called covariant-contravariant refinement (modality). It is inspired by process calculi and not by epistemic modalities and updates. 
As in \cite{bozzellietal.inf:2014}, these modalities are parameterized by the subgroups of the set of all agents for which the refinement (resp.\ simulation) relation only needs to satisfy \textbf{back} (resp.\ \textbf{forth}). In this work, we only refine or simulate by the set of all agents, thus getting a more succinct formalization \textbf{SML}, presented in Section~\ref{sec:axiomatization}. 
In that section we furthermore present the axiomatization \textbf{ROSML}, that is truly novel.

An aspect of \textbf{SML} that might be criticized is the somewhat extra-logical (or at least uncommon in a Hilbert-style axiomatization) requirement of consistency of the formulas appearing in the cover. 
Proofs and algorithms are different beasts, but it is an aesthetic principle that determining whether a rule or axiom is applicable in a proof should be self-evident; at least in the sense that it can be determined in linear time.
We get rid of this problem in the \textbf{ROSML} axiomatization by a more complex reduction also involving origin modalities and refinement modalities but having no extra-logical requirements. 


Just as the refinement relation, the simulation relation is also transitive, reflexive, and confluent, and again we have atomicity (the maximal simulation brings you in the mutual factual ignorance model). We therefore  have for the simulation modality the following validities.
\begin{restatable}{proposition}{propSIMform} 
The following are valid: 
\begin{enumerate}
\item $\smlbox \phi \imp \phi$ \quad ({\bf T}$^\simulates$)
\item $\smldia \smldia \phi \imp \smldia \phi$ \quad ({\bf 4}$^\simulates$) 
\item $\smldia \smlbox \phi \imp \smlbox \smldia \phi$ \quad (Church-Rosser) / {\bf CR}$^\simulates$
\item $\smlbox \smldia \phi \imp \smldia \smlbox \phi$ \quad (McKinsey) / {\bf MK}$^\simulates$ 
\end{enumerate} 
\end{restatable}
\begin{proof}
Let $M = (S,R,V)$ and $s \in S$ be given. 

\begin{enumerate}
\item Assume $M_s \models \smlbox \phi$. As $M \simulates M$ (the relation $\simulates$ is reflexive), we directly obtain that $M_s \models \phi$. Therefore $\models\smlbox \phi \imp \phi$.
\item Assume $M_s \models \smldia \smldia \phi$. Then there are $M'$ and $M''$ with $s' \in \domain(M')$ and $s'' \in \domain(M'')$ such that $M_s \simulates M'_{s'}$ and $M'_{s'} \simulates M''_{s''}$, and such that $M''_{s''} \models \phi$. From $M_s \simulates M'_{s'}$ and $M'_{s'} \simulates M''_{s''}$ we directly obtain $M_s \simulates M''_{s''}$ as the simulation relation is transitive. Therefore $M_s \models \smldia\phi$, so that $\models \smldia \smldia \phi \imp \smldia \phi$.
\item 
Assume that $M_s \models \langle \rightrightarrows \rangle [\rightrightarrows]\varphi$ 
Then, by the definition of semantics, there is an $M'_{s'}$ such that $M_s \rightrightarrows M'_{s'}$ and  $M'_{s'} \models [\rightrightarrows]\varphi$. The latter implies, again by the definition of semantics, that $M^\circ _{s^\circ} \models \varphi$. Since $M^\circ$ is a simulation of every model, we have that  $M^{''}_{s^{''}} \models \langle \rightrightarrows \rangle\varphi$ for an arbitrary $M^{''}_{s^{''}}$ such that $M_s \rightrightarrows M^{''}_{s^{''}}$. Since $M^{''}_{s^{''}}$ was arbitrary, the definition of the semantics implies that $M_s \models [\rightrightarrows] \langle \rightrightarrows \rangle \varphi$.

\item  
Assume that $M_s \models [\rightrightarrows] \langle \rightrightarrows \rangle \varphi$. 
From the definition of semantics and the reflexivity of the simulation relation we get $M^\circ_{s^\circ} \models \varphi$. As the mutual factual ignorance model is a simulation of itself, we trivially get $M^\circ_{s^\circ} \models [\rightrightarrows]\varphi$. Finally, since $M_s \rightrightarrows M^\circ_{s^\circ}$, we have that $M_s \models \langle \rightrightarrows \rangle [\rightrightarrows]\varphi$. 
\end{enumerate}
\end{proof}

While refinement preserves the positive formulas, or hard knowledge \cite{bozzellietal.inf:2014},
simulation preserves ignorance. Consider the following, \textit{negative}, fragment $\lang^{\Box}_-$ of $\lang^{\Box}$: 
\[ \phi ::= p \mid \neg p \mid \phi \vel \psi \mid \phi \et \phi \mid \Dia_a \phi \]

\begin{proposition}
For all $\phi \in \lang^{\Box\simulates}_-$ and models $M_s$, if $M_s \models \phi$ and $M_s \simulates M'_{s'}$, then $M'_{s'} \models \phi$. 
\end{proposition}
\begin{proof}
    The proof is by induction on the structure of $\varphi$, where the base case, $\varphi = p$ and $\varphi = \lnot p$, is immediate from the fact that simulations preserve the valuation of the original states. 
    Boolean cases follow by the induction hypothesis. 

    \textit{Case} $\varphi = \Diamond_a \psi$. Assume that $M_s \models \Diamond_a \psi$. This means that there is an $a$-reachable state $t$ such that $M_t \models \psi$. By the induction hypothesis, we have that $M'_{t'} \models \psi$. Finally, since there is a simulation $(s,s')\in Z$, we have that $(s',t') \in R'(a)$ with $M'_{t'} \models \psi$, and hence $M'_{s'} \models \Diamond_a \psi$.
\end{proof}

\paragraph{Mutual Factual Ignorance}  

The concept of mutual factual ignorance works as an epistemic reset; it describes a state where no agent has any {\em factual} knowledge, and this lack of factual knowledge is {\em common} knowledge.
Such a state is refered to as {\em mutual factual ignorance}. Reaching this state of information can be seen as a form of epistemic update: ``If we were all to discount all our factual knowledge, then $\phi$ would be true''.

As the mutual factual ignorance model $M^\circ$ serves as a kind of \textit{tabula rasa} for agents, any model where agents have some factual information is a refinement of $M^\circ$.

\begin{proposition}\label{prop:MI2ref}
Every model is a refinement of mutual factual ignorance model $M^\circ$.
\end{proposition}
\begin{proof}
    Take an arbitrary model $M = (S,R,V)$. From the definition of $M^\circ$, it is clear that we can identify every $s\in S$ with $V(s) \in S^\circ$ satisfying the same propositional variables. 
    Having the refinement relation $Z = \{(V(s),s) \mid s\in S\}$, it is easy to check that $M$ is a refinement of $M^\circ$.  
\end{proof}

As a direct corollary of Proposition~\ref{prop:MI2ref}, we have the fact that if a formula is true in all refinements of the mutual factual ignorance model, then the formula is valid.

\begin{proposition}
Given some $\phi$, if for every $P'\subseteq P$, $M^\circ_{P'}\models \rmlbox \phi$, then $\phi$ is valid.
\end{proposition}

All relations $R^\circ_a$ in the mutual factual ignorance model are total, and therefore equivalence relations. In the axiomatization \textbf{OML} of OML, this appears from the axioms involving the origin modality such as $[\circ](\Box_a\phi \imp \phi)$, that formalizes that ``given mutual factual ignorance, epistemic modalities $\Box_a$ are \textbf{S5}''.



%
%
%
%

\section{Axiomatizations}\label{sec:axiomatization}

In this section we present the axiomatizations side by side and show their soundness and completeness.

\subsection{Proof Systems} \label{sbsect:proof}

\paragraph{Modal Logic}
Let us start with the basic axiomatization for the language $\lang^{\Box}$, which we call {\bf ML}:
\[\begin{array}{ll}
{\bf Prop} & \text{all substitution instances of tautologies of propositional logic} \\
{\bf K} & \Box_a(\phi \rightarrow \psi)\rightarrow (\Box_a\phi\rightarrow \Box_a\psi)\\
{\bf MP} & \text{from } \phi\rightarrow \psi \text{ and } \phi \text{ infer }  \psi\\
{\bf N} & \text{from } \phi \text{ infer } \Box_a\phi\\
{\bf RE} & \text{from } \chi\eq\psi \text{ infer } \phi[\chi/p] \eq \phi[\psi/p]\\
\end{array}
\]
This forms a foundation for the axioms that follow and its soundness and completeness 
for the class of all Kripke models is well established (see for example \cite{blackburnetal:2001}).

\paragraph{Refinement Modal Logic}

Let ${\bf REF}$ be the following set of axioms:
\[
\begin{array}{ll}
{\bf RQ1} & 
\rmldia \phi_0 \leftrightarrow \phi_0 \hspace{4cm} \hfill \text{where } \phi_0\in \lang_0\\
{\bf RQ2} & 
\rmldia (\phi \vee \psi)\leftrightarrow (\rmldia \phi \vee \rmldia \psi)\\
{\bf RQ3} & 
\rmldia (\phi_0 \wedge \phi)\leftrightarrow (\phi_0\wedge \rmldia\phi) \hfill \text{where } \phi_0\in \lang_0\\
{\bf RQ4} & 
\rmldia\Et_{a\in A}\cover_a \Phi_a\iff\bigwedge_{a\in\agents}\bigwedge_{\phi\in\Phi_a}\lozenge_a\rmldia\phi\\
\end{array} \]
The axiomatization for refinement modal logic is then {\bf RML} = {\bf ML} + {\bf REF}, where axioms and rules of \textbf{ML} are applied to $\lang^{\Box\refines}$. It is based on the axiomatization of arbitrary arrow update model logic in \cite{hvdetal.aus:2020}. 

As an example, 
we show that 
rule ${\bf AR}$
: `$\text{from } \phi \rightarrow \psi \text{ infer } \rmldia \phi\rightarrow \rmldia \psi$', is derivable.

\begin{center}
\begin{tabular}{lllclll}
1. & $\phi \rightarrow \psi$ & \text{Given}&\quad&
2. & $(\phi\vee \psi)\leftrightarrow \psi$ & \textbf{Prop} and 1\\
3. & $\rmldia (\phi\vee\psi) \leftrightarrow \rmldia \psi$ & \textbf{RE}&\quad&
4. & $\rmldia (\phi\vee\psi) \leftrightarrow (\rmldia\phi \vee \rmldia\psi)$ & \textbf{RQ2}\\
5. & $(\rmldia\phi \vee \rmldia\psi)\leftrightarrow \rmldia\psi$ & \textbf{MP} and 3,4&\quad&
6. & $\rmldia\phi\rightarrow \rmldia\psi$ & \textbf{Prop} and 5
\end{tabular}
\end{center}

\paragraph{Origin Modal Logic}

Origin modal logic is quite novel, and correspondingly, has a novel set of axioms.
For origin modal logic we start with the set {\bf ORI}, given by
\[\begin{array}{ll}
{\bf O1} & \originbox \phi_0 \leftrightarrow \phi_0 \hspace{4cm} \hfill \text{where } \phi_0\in \lang_0\\
{\bf OT} & \originbox (\Box_a\phi\rightarrow \phi)\\
{\bf O5} & \originbox (\lozenge_a\phi\rightarrow \Box_a\lozenge_a\phi)\\
{\bf OExch} & \originbox (\Box_a\phi \rightarrow \Box_b\phi)\\ 
{\bf OFull} & \originbox \lozenge_a\phi \quad \hfill \text{ where } \phi \text{ is of the form } \bigwedge_{p\in Q_1} p \wedge \bigwedge_{p\in Q_2}\neg p \text{ with } Q_1\cap Q_2=\emptyset\\
{\bf ODual} & \originbox \neg \phi \leftrightarrow \neg \originbox\phi\\
{\bf ODisj} & \originbox (\phi\vee\psi) \leftrightarrow (\originbox \phi \vee \originbox \psi)\\
&\\
{\bf OMP} & \text{from } \originbox (\phi\rightarrow \psi) \text{ and } \originbox \phi \text{ infer } \originbox \psi\\
{\bf ON} & \text{from } \originbox \phi \text{ infer } \originbox\Box_a\phi\\
\end{array} \]
and obtain {\bf OML} = {\bf ML} + {\bf ORI}, where axioms and rules of {\bf ML} are applied to $\lang^{\Box\circ}$.

Observe that necessitation for the origin modality, i.e.`from $\phi$ infer $\originbox \phi$', is derivable in \textbf{OML}.

\begin{center}
\begin{tabular}{lllclll}
1. & $\phi$ &Given&\quad&
2. & $\phi\leftrightarrow (p \lor \lnot p)$ & 1\\
3. & $\originbox p \lor \lnot \originbox p$ & \bf{Prop}&\quad&
4. & $\originbox p \lor \originbox \lnot p$ & \textbf{ODual} and 3\\
5. & $\originbox(p \lor \lnot p)$ & \textbf{ODisj} and 4&\quad&
6. & $\originbox \phi$ & \textbf{RE} and 5
\end{tabular}
\end{center}



\paragraph{Simulation Modal Logic}
We report two ways to axiomatize the logic with simulation quantifiers. First, one can use refinement quantifiers and the origin modality to ensure that all formulas in a cover set are satisfiable in a refinement of the mutual factual ignorance model (see axiom \textbf{SQ4} below). 
Then, let {\bf SIM} be 
\[\begin{array}{ll}
{\bf SQ1} & \smldia \phi_0 \leftrightarrow \phi_0 \hfill\text{ where }\phi_0\in \lang_0\\
{\bf SQ2} & \smldia (\phi\vee\psi)\leftrightarrow (\smldia \phi\vee\smldia\psi)\\
{\bf SQ3} & \smldia (\phi_0\wedge \phi) \leftrightarrow (\phi_0\wedge \smldia\phi)\hfill\text{ where }\phi_0\in \lang_0\\
{\bf SQ4} & \smldia\Et_{a\in A}\cover_a\Phi_a\iff \bigwedge_{a\in\agents} (\knows_a\bigvee_{\phi\in\Phi_a}\smldia\phi 
                   \land \originbox \bigwedge_{\phi\in\Phi_a}\suspects_a\rmldia\phi) \\
\end{array} \]
and {\bf ROSML} = {\bf ML} + {\bf REF} + {\bf ORI} + {\bf SIM}, where all of the axioms and rules are applied to $\lang$.

Another way to axiomatize simulation quantifiers is to explicitly require all the formulas in a cover set to be consistent. In this way, we do not need refinement quantifiers and the origin modality, but we have to pay for this with the higher complexity of determining whether or not an axiom may be faithfully applied (in this case, \textbf{PSPACE} for the satisfiability problem of modal logic K \cite{blackburnetal:2001}).
So, 
let {\bf SIM}$_{\mathsf{cons}}$ be {\bf SIM} where we replace the axiom {\bf SQ4} by
\[\begin{array}{ll}
{\bf SQ4}_{\mathsf{cons}} & \smldia\Et_{a\in A}\cover_a\Phi_a\iff \bigwedge_{a\in\agents}\knows_a\bigvee_{\phi\in\Phi_a}\smldia\phi \qquad \hfill \text{where all $\phi$ are consistent $\lang^\Box$ formulas}
\end{array} \]
then {\bf SML} = {\bf ML} + {\bf SIM}$_{\mathsf{cons}}$, where axioms and rules of {\bf ML} are applied to $\lang^{\Box\simulates}$. Axiom ${\bf SQ4}_{\mathsf{cons}}$ is inspired by axiom {\bf CCRKco2} in \cite{xingetal:2019}, that also has the consistency requirement. 


Note the similarity between \textbf{SML} and \textbf{RML}. 
All but axiom ${\bf SQ4}_{\mathsf{cons}}$ simply replaced $\rmldia$ modalities by $\smldia$ modalities. 
Intuitively the difference is clear: 
{\bf RQ4} says that if after refinement we can cover all the formulas in $\Phi_a$, then before the refinement all those must have been possible but it may no longer be a cover, so we lose the $\Box$. Dually, ${\bf SQ4}_{\mathsf{cons}}$ says that if after simulation we can cover all the formulas in $\Phi_a$, then before the simulation necessarily one of those must already have been reachable by a simulation. 
However, there is also a consistency requirement there, 
without which the axiom is unsound. Consider a single agent. We then get \[ \smldia(\bigwedge_{\phi \in \Phi}\Dia\phi\wedge \Box \bigvee_{\phi \in \Phi}\phi)\leftrightarrow \Box \bigvee_{\phi \in \Phi}\smldia \phi. \] The simplest counterexample to soundness is the set $\Phi = \{\top,\bot\}$, 
noting that $\Box\top \eq \top$. We get \[ \smldia(\Dia\top \et \Dia\bot)\leftrightarrow \Box (\smldia\top \vel \smldia\bot). \] Since the right-hand side of the equivalence requires that just one of the elements of $\Phi$ needs to be true at every reachable state, it suffices to observe that $\top$ is true everywhere, so that the disjunction is then also true. But then this also requires the left-hand side of the equivalence to be true: there is a simulation where $\Phi$ covers the successors, so there must be an accessible world where $\bot$ is true, which is an absurdity.

\subsection{Soundness}
In this section, we show that each of the axiom systems is sound for its logic. Soundness of {\bf ML} is well known, so we will not comment further on it here. Let us therefore consider the other axiom systems. 
\begin{restatable}{proposition}{propRMLsound}
\label{prop:RMLsoundness}
The axiomatization {\bf RML} is sound  for $\lang^{\Box\refines}$.
\end{restatable}
\begin{proof}
We have ${\bf RML} = {\bf ML} + {\bf REF}$, and soundness of the axioms and rules of {\bf ML} is well known. What remains to show is that {\bf REF} is sound.

Refinements satisfy {\bf atoms}, and therefore do not change the truth value of propositional formulas. Hence $\models \rmldia\phi_0\leftrightarrow \phi_0 $, so {\bf RQ1} is sound.

Next, note that $\rmldia$ is a ``diamond-like'' operator, in that $\rmldia\phi$ holds iff there is at least one refinement after which $\phi$ is true. Because it is diamond-like, $\rmldia$ distributes over disjunctions, so $\models\rmldia (\phi \vee \psi)\leftrightarrow (\rmldia \phi \vee \rmldia \psi)$, which means that {\bf RQ2} is sound.

For {\bf RQ3}, we start with the left-to-right direction, so suppose that $M_s\models \rmldia(\phi_0\wedge \phi)$, where $\phi_0\in \lang_0$. Then there is some $M'_{s'}$ such that $M_s\refines M'_{s'}$ and $M'_{s'}\models \phi_0\wedge \phi$. As such, we have $M'_{s'}\models \phi_0$ and $M'_{s'}\models \phi$, which imply that $M_s\models \rmldia \phi_0$ and $M_s\models \rmldia \phi$, respectively. Furthermore, as shown above, $M_s\models \rmldia\phi_0\leftrightarrow \phi_0$. Hence $\models \rmldia(\phi_0\wedge\phi)\rightarrow (\phi_0\wedge\rmldia\phi)$.

Suppose then, for the right-to-left direction, that $M_s\models \phi_0\wedge \rmldia\phi$. Then $M_s\models \rmldia\phi$, so there is some $M'_{s'}$ such that $M_s\refines M'_{s'}$ and $M'_{s'}\models \phi$. Because every refinement satisfies {\bf atoms}, $M_s\refines M'_{s'}$ together with $M_s\models \phi_0$ implies that $M'_{s'}\models \phi_0$. So $M'_{s'}\models \phi_0\wedge \phi$, and therefore $M_s\models \rmldia (\phi_0\wedge \phi)$. So $\models ( \phi_0\wedge \rmldia\phi)\rightarrow \rmldia(\phi_0\wedge\phi)$. Together with the previously shown $\models \rmldia(\phi_0\wedge\phi)\rightarrow (\phi_0\wedge\rmldia\phi)$ this implies that $\models \rmldia(\phi_0\wedge\phi)\leftrightarrow (\phi_0\wedge\rmldia\phi)$, so {\bf RQ3} is sound.

For {\bf RQ4} we also start with the left-to-right direction, so suppose that $M_{s_1}\models \rmldia\Et_{a\in A}\cover_a\Phi_a$. Then there is some $M'_{s_1'}$ such that $M_{s_1}\refines M'_{s_1'}$ and $M'_{s_1'}\models \Et_{a\in A}\cover_a\Phi_a$. This implies that, in particular, for every $a\in A$ and $\phi\in \Phi_a$, there is some $s_2'$ such that $(s_1',s_2')\in R'_a$ and $M'_{s_2'}\models \phi$. Because $M_{s_1}\refines M'_{s_1'}$ and every refinement satisfies the {\bf back} condition, there must be some $s_2$ such that $(s_1,s_2)\in R_a$ and $M_{s_2}\refines M'_{s_2'}$. We then have $M_{s_2}\models \rmldia \phi$, and therefore $M_{s_1}\models \lozenge_a\smldia\phi$. This holds for every $a\in A$ and $\phi\in \Phi_a$, so $M_{s_1}\models \Et_{a\in A}\Et_{\phi\in \Phi_a}\lozenge_a\rmldia\phi$. We have now shown that $\models \rmldia\Et_{a\in A}\cover_a\Phi_a\rightarrow \Et_{a\in A}\Et_{\phi\in \Phi_a}\lozenge_a\rmldia\phi$.

For the right-to-left direction, suppose that $M_{s_0}\models \Et_{a\in A}\Et_{\phi\in \Phi_a}\lozenge_a\rmldia\phi$. So for every $a\in A$ and $\phi\in\Phi_a$ there is some $s_\phi$ such that $(s_0,s_\phi)\in R_a$ and $M_{s_\phi}\models\rmldia\phi$. So there are $M^\phi=(S^\phi,R^\phi,V^\phi)$ and $s'_{\phi}$ such that $M_{s_\phi}\refines M^\phi_{s_\phi'}$ and $M^\phi_{s_\phi'}\models\phi$. Now, let $M'$ be the disjoint union of all $M^\phi$, with one additional state $s_0'$, where we take $V'(s_0')=V(s_0)$, and additional $a$-edges from $s_0'$ to $s_\phi'$ for all $a\in A$ and $\phi\in\Phi_a$.
Formally, $M'=(S',R',V')$, where
\begin{itemize} 
\item $S' = \{s_0'\}\cup \bigcup_{a\in A}\bigcup_{\phi\in \Phi_a}\{(s',a,\phi)\mid s'\in S_\phi\}$,
\item $R'_a = \{(s_0',(s'_\phi,a,\phi))\mid a\in A, \phi\in \Phi_A\}\cup \bigcup_{a\in A}\bigcup_{\phi\in \Phi_a}\{((s_1',a,\phi), (s_2',a,\phi))\mid (s_1',s_2')\in R_a^\phi\}$,
\item $V'(s') = \left\{ \begin{array}{ll}V(s_0) & \text{if }s'=s_0'\\ V^\phi(t')&\text{if } s' = (t',a,\phi)\end{array}\right.$
\end{itemize}
See Figure~\ref{fig:M'-construction} for an example of this construction.

For every $\phi\in \Phi_a$, there is some $s'$ such that $(s_0',s')\in R'_a$ and $M'_{s'}\models\phi$, namely $s'=(s'_\phi,a,\phi)$. Conversely, for every $s'$ such that $(s_0',s')\in R'_a$ there is some $\phi\in \Phi_a$ such that $M'_{s'}\models \phi$, because $(s_0',s')\in R'_a$ implies that $s'$ is of the form $s'=(s_\phi',a,\phi)$. This implies that $M'_{s_0'}\models \Et_{a\in A}\cover_a\Phi_a$.

Furthermore, for every $a\in A$ and $\phi\in\Phi_a$ we have $M_{s_\phi}\refines M^\phi_{s_\phi'}\models\phi$. Let $Z_\phi$ be the witnessing refinement, and let $Z$ be the disjoint union of all $Z_\phi$ plus the pair $(s_0,s_0')$. Formally, $Z_\phi = (s_0,s_0')\cup \bigcup_{a\in A}\bigcup_{\phi\in \Phi_a}\{(s,(s',a,\phi))\mid (s,s')\in Z_\phi\}$. We will show that this $Z$ is a refinement.

Take any $(s,s')\in Z$. We consider two cases. Firstly, if $s'=s_0'$, then $s=s_0$. We chose $V'(s_0')=V(s_0)$, so in this case {\bf atoms} is satisfied. Secondly, if $s'\not = s_0'$ then $s'$ is of the form $s'=(t',a,\phi)$, and we have $(s,t')\in Z^\phi$, which implies that $V(s)=V^\phi(t')$ because $Z^\phi$ satisfies {\bf atoms}. We took $V'(s')=V^{\phi}(t')$, so from $V(s)=V^\phi(t')$ it follows that $V(s)=V'(s')$, so {\bf atoms} holds for $s$ and $s'$. These two cases are exhaustive, so {\bf atoms} is satisfied for $Z$.

Then take any $(s_1,s_1')\in Z$ and $(s_1',s_2')\in R'_a$. Again, we consider two cases. Firstly, if $s_1'=s_0'$, then $s_1=s_0$ and $s_2'$ is of the form $s_2'=(s_\phi',a,\phi)$. We have $(s_\phi,s'_\phi)\in Z^\phi$ which, by the construction of $Z$, implies that $(s_2,s_2')=(s_\phi,(s_\phi',a,\phi))\in Z$. Furthermore, $(s_1,s_2)=(s_0,s_\phi)\in R_a$, so {\bf back} is satisfied in this case. Secondly, if $s_1'\not = s_0'$, then $s_1'$ and $s_2'$ are of the form $s_1'=(t_1',b,\phi)$ and $s_2'=(t_2',b,\phi)$, for some $b$ and $\phi$. We then also have, by the construction of $Z$ and $M'$, that $(t_1',t_2')\in R^\phi$ and $(s_1,t_1')\in Z^\phi$. The {\bf back} condition for $Z^\phi$ therefore implies that there is some $s_2$ such that $(s_1,s_2)\in R_a$ and $(s_2,t_2')\in Z^\phi$. This implies that $(s_1,s_2)\in R_a$ and $(s_2,s_2')=(s_2,(t_2',b,\phi))\in Z$, so {\bf back} is satisfied in this case as well. The two cases are exhaustive, so {\bf back} is satisfied for $Z$.

$Z$ is therefore a refinement, which implies that $M_{s_0}\refines M'_{s_0'}$. Together with the previous conclusion that $M'_{s_0'}\models \Et_{a\in A}\cover_a\Phi_a$, this implies that $M_{s_0}\models \rmldia\Et_{a\in A}\cover_a\Phi_a$. We have now shown that \[\models \Et_{a\in A}\Et_{\phi\in \Phi_a}\lozenge_a\rmldia\phi \rightarrow \rmldia\Et_{a\in A}\cover_a\Phi_a,\] which together with the previously shown \[\models \rmldia\Et_{a\in A}\cover_a\Phi_a\rightarrow \Et_{a\in A}\Et_{\phi\in \Phi_a}\lozenge_a\rmldia\phi\] implies that  \[\models \rmldia\Et_{a\in A}\cover_a\Phi_a\leftrightarrow \Et_{a\in A}\Et_{\phi\in \Phi_a}\lozenge_a\rmldia\phi.\] So {\bf RQ4} is sound.
\end{proof}
\begin{figure}
\begin{center}
\begin{tikzpicture}[scale=0.9, transform shape]
\draw[rounded corners=1cm] (-1.5,1) -- (3.5,1) -- (3.5,-3.2) -- (-1.5,-3.2) -- cycle;
\node at (1,-3.5) {$M$};
\node[label=above:$s_0$] (s) at (1,0) {$\bullet$};

\node[label=below:$s_{\phi_1}$] (sphi1) at (-0.5,-2.2) {$\bullet$};
\node[label=below:$s_{\phi_2}$] (sphi2) at (0.5,-2.2) {$\bullet$};
\node[label=below:$s_{\phi_3}$] (spsi1) at (1.5,-2.2) {$\bullet$};
\node[label=below:$s_{\phi_4}$] (spsi2) at (2.5,-2.2) {$\bullet$};

\draw[->] (s) -- (sphi1) node[midway,left] {$a$};
\draw[->] (s) -- (sphi2) node[midway,left] {$a$};
\draw[->] (s) -- (spsi1) node[midway,right] {$b$};
\draw[->] (s) -- (spsi2) node[midway,right] {$b$};

\draw[dashed] (4,1.2) -- (4,-5);

\draw[rounded corners=1cm] (4.8,1) -- (13.2,1) -- (13.2,-4.4) -- (4.8,-4.4) -- cycle;
\node at (9,-4.7) {$M'$};

\node[label=above:$s_0'$] (s0') at (9,0) {$\bullet$};

\node[label=below:$s_{\phi_1}'$] (sphi1') at (6,-2.2) {$\bullet$};
\node[draw,ellipse,minimum width=1.5cm,minimum height=2.5cm,label=below:$M^{\phi_1}$] (M1) at (6,-2.5){};
\node[label={below=2pt}:$s_{\phi_2}'$] (sphi2') at (8,-2.2) {$\bullet$};
\node[draw,ellipse,minimum width=1.5cm,minimum height=2.5cm,label=below:$M^{\phi_2}$] (M2) at (8,-2.5){};
\node[label=below:$s_{\phi_3}'$] (spsi1') at (10,-2.2) {$\bullet$};
\node[draw,ellipse,minimum width=1.5cm,minimum height=2.5cm,label=below:$M^{\phi_3}$] (M3) at (10,-2.5){};
\node[label=below:$s_{\phi_4}'$] (spsi2') at (12,-2.2) {$\bullet$};
\node[draw,ellipse,minimum width=1.5cm,minimum height=2.5cm,label=below:$M^{\phi_4}$] (M4) at (12,-2.5){};

\draw[->] (s0') -- (sphi1') node[midway,left] {$a$};
\draw[->] (s0') -- (sphi2') node[midway,right] {$a$};
\draw[->] (s0') -- (spsi1') node[midway,left] {$b$};
\draw[->] (s0') -- (spsi2') node[midway,right] {$b$};

\draw[dotted] (s) edge[bend left] (s0');
\draw[dotted] (sphi1) edge[bend right] (sphi1');
\draw[dotted] (sphi2) edge[bend right] (sphi2');
\draw[dotted] (spsi1) edge[bend right] (spsi1');
\draw[dotted] (spsi2) edge[bend right] (spsi2');
\end{tikzpicture}
\end{center}

\caption{Schematic drawing of an example of $M'$ as in the proof of Proposition~\ref{prop:RMLsoundness}. In this example, $A=\{a,b\}$, $\Phi_a=\{\phi_1,\phi_2\}$ and $\Phi_b=\{\phi_3,\phi_4\}$. Dotted lines represent the relation $Z$. The state $s_0'$ has exactly the successors $s_{\phi_1}', \cdots, s_{\phi_4}'$, the state $s_0$ has at least the successors $s_{\phi_1},\cdots, s_{\phi_4}$ but may have more. Hence $Z$ is a refinement but not, in general, a bisimulation.}
\label{fig:M'-construction}
\end{figure}
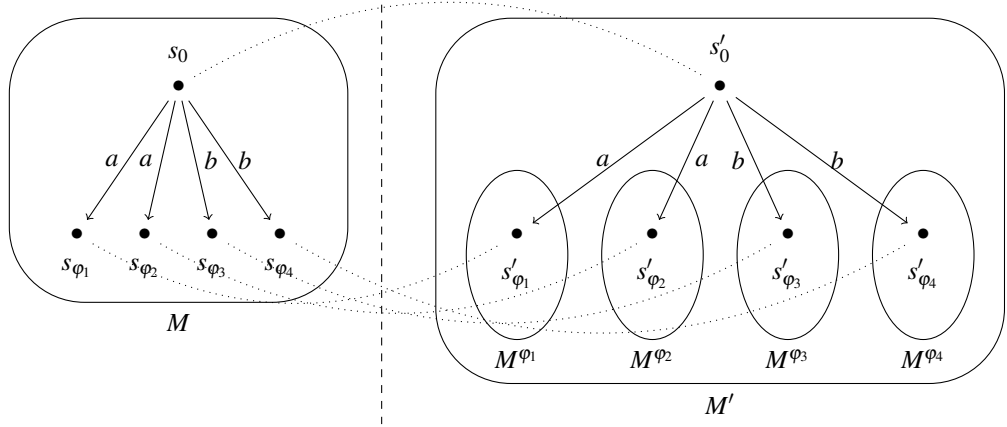

The soundness of \textbf{RQ4} is shown rather differently than that of the similar axiom \textbf{A4} in \cite{hvdetal.aus:2020}, involving quantifiers over arrow updates. However, the proof in the current setting has nice duality with that of the soundness of \textbf{SQ4}, coming up (almost) next, in Proposition~\ref{prop:SMLsoundness}.

\begin{restatable}{proposition}{propOMLsound}
\label{prop:OMLsoundness}
The axiomatization {\bf OML} is sound for $\lang^{\Box\circ}$.
\end{restatable}
\begin{proof}
We have ${\bf OML} = {\bf ML} + {\bf ORI}$, and the soundness of {\bf ML} is well known, what remains to show is that {\bf ORI} is sound.

To begin with, for every $M_s$ and every $\phi_0\in \lang_0$, we have $M_s\models \phi_0\Leftrightarrow M^\circ_{V(s)}\models \phi_0 \Leftrightarrow M_s\models \originbox\phi_0$, so $\models \originbox\phi_0\leftrightarrow \phi_0$, which means that {\bf O1} is sound.

The model $M^\circ$ is an S5-model, i.e., all relations in $M^\circ$ are equivalence relations. It follows that, for every $t\in S^\circ$, $M^\circ_{t}\models \square_a\phi\rightarrow \phi$ and $M^\circ_{t}\models \lozenge_a\phi\rightarrow \square_a\lozenge_a\phi$. This implies that, for all $M_s$, we have $M^\circ_{V(s)}\models \square_a\phi\rightarrow \phi$ and $M^\circ_{V(s)}\models \lozenge_a\phi\rightarrow \square_a\lozenge_a\phi$, and hence $M_s\models \originbox (\square_a\phi\rightarrow \phi)$ and $M_s\models \originbox(\lozenge_a\phi\rightarrow \square_a\lozenge_a\phi)$. As this is true for every $M_s$, we have $\models \originbox(\square_a\phi\rightarrow \phi)$ and $\models \originbox(\lozenge_a\phi\rightarrow \square_a\lozenge_a\phi)$, so {\bf OT} and {\bf O5} are sound.

Furthermore, all agents $a,b\in A$ have the same accessibility relation on $M^\circ$, namely $R^\circ_a = R^\circ_b = S^\circ\times S^\circ$. This implies that $M^\circ_{t}\models \square_a\phi\rightarrow\square_b\phi$, for every $t$, and therefore $\models \originbox (\square_a\phi\rightarrow\square_b\phi)$. So {\bf OExch} is sound.

Next, consider $Q_1,Q_2\subseteq P$ such that $Q_1\cap Q_2=\emptyset$. Then there is at least one valuation $t$ such that $M^\circ_{t}\models \Et_{p\in Q_1}p \wedge \Et_{p\in Q_2}\neg p$. Because $R^\circ_a = S^\circ\times S^\circ$, it follows that $M^\circ_{t'}\models \lozenge_a(\Et_{p\in Q_1}p \wedge \Et_{p\in Q_2}\neg p)$ for all $t'$, and therefore also that $\models \originbox\lozenge_a (\Et_{p\in Q_1}p \wedge \Et_{p\in Q_2}\neg p)$. So {\bf OFull} is sound.

We have $M_s\models \originbox \neg \phi \Leftrightarrow M^\circ_{V(s)}\models \neg \phi \Leftrightarrow M^\circ_{V(s)}\not\models \phi\Leftrightarrow M_s\not\models \originbox\phi\Leftrightarrow M_s\models\neg\originbox \phi$. It follows that $\models \originbox\neg \phi \leftrightarrow \neg \originbox\phi$, so {\bf ODual} is sound.

Similarly, $M_s\models \originbox (\phi\vee\psi) \Leftrightarrow M^\circ_{V(s)}\models \phi\vee\psi \Leftrightarrow (M^\circ_{V(s)}\models \phi \text{ or } M^\circ_{V(s)}\models \psi )\Leftrightarrow (M_s\models \originbox \phi \text{ or } M_s\models\originbox \psi )\Leftrightarrow M_s\models \originbox \phi\vee\originbox\psi$. It follows that $\models \originbox(\phi\vee\psi)\leftrightarrow (\originbox\phi\vee\originbox\psi)$, so {\bf ODisj} is sound.

This leaves the two rules, we show that they preserve validity. So suppose that $\models \originbox(\phi\rightarrow \psi)$ and $\models \originbox\psi$. Then for every $M_s$, we have $M_s\models \originbox (\phi\rightarrow \psi)$ and $M_s\models \originbox\psi$. So $M^\circ_{V(s)}\models \phi\rightarrow \psi$ and $M^\circ_{V(s)}\models \phi$, which implies that $M^\circ_{V(s)}\models \psi$ and therefore $M_s\models \originbox\psi$. Since this holds for every $M_s$, we then have $\models \originbox\psi$, proving the soundness of {\bf OMP}.

Finally, suppose that $\models \originbox \phi$. Then for every $M_s$, we have $M_s\models \originbox\phi$ and therefore $M^\circ_{V(s)}\models \phi$. Since this holds for every $M_s$, and therefore also every $V(s)$, we also have $M^\circ_{V(s')}\models \phi$ for every $V(s')$, which implies that $M^\circ_{V(s)}\models \square_a\phi$, and therefore $M_s\models \originbox \square_a\phi$. Since this holds for every $M_s$, we then have $\models \originbox \square_a\phi$. So {\bf ON} is sound.
\end{proof}

\begin{restatable}{proposition}{propSMLsound}
\label{prop:SMLsoundness}
The axiomatization {\bf SML} is sound  for $\lang^{\Box\simulates}$.
\end{restatable}
\begin{proof}

We have ${\bf SML} = {\bf ML} + {\bf SIM_{cons}}$. Soundness of {\bf ML} is well known, what remains to be shown is the soundness of ${\bf SIM_{cons}}$. Axioms {\bf SQ1}, {\bf SQ2}, and {\bf SQ3} are sound for the same reasons that {\bf RQ1}, {\bf RQ2} and {\bf RQ3} are: preservation of propositional logic under simulations for {\bf SQ1} and {\bf SQ3}, and $\smldia$ being ``diamond-like'' for {\bf SQ2}. This leaves ${\bf SQ4_{cons}}$.

For the left-to-right direction, suppose that $M_s\models \smldia\Et_{a\in A}\cover_a\Phi_a$.\footnote{Note that we are not, at this point, explicitly assuming that all $\phi\in \Phi_a$ are consistent. This is because $M_s\models \smldia\Et_{a\in A}\cover_a\Phi_a$ implies that all $\phi$ are consistent. The separate consistency condition in ${\bf SQ4_{cons}}$ is only required for the right-to-left direction.} Then there is some $M'_{s'}$ such that $M_s\simulates M'_{s'}$ and $M'_{s'}\models \Et_{a\in A}\cover_a\Phi_a$. Take any $t$ such that $(s,t)\in R_a$. Because $M_s\simulates M'_{s'}$, there is some $t'$ such that $(s',t')\in R'_a$ and $M_{t}\simulates M'_{t'}$.

Now, because $M'_{s'}\models \Et_{a\in A}\cover_a\Phi_a$, there is some $\phi\in \Phi_a$ such that $M'_{t'}\models\phi$, and therefore $M_{t}\models \smldia\phi$. This holds for every $(s,t)\in R_a$ and every $a\in A$, so we have $M_s\models\Et_{a\in A}\square_a\bigvee_{\phi\in \Phi_a}\smldia\phi$. We have now shown that
\[\models \smldia \Et_{a\in A}\cover_a\Phi_a \rightarrow \Et_{a\in A}\square_a\bigvee_{\phi\in \Phi_a}\smldia \phi.\]

For right-to-left direction, suppose that $M_{s_0}\models \Et_{a\in A}\square_a\bigvee_{\phi\in \Phi_a}\smldia\phi$, where all $\phi$ are consistent. Take any $(s,t)\in R_a$. Because $M_s\models \square_a\bigvee_{\phi\in \Phi_a}\smldia\phi$, there is some $\phi\in \Phi_a$ such that $M_{t}\models \smldia\phi$. Then there are some model $M^{t}$ and state $t'$ of $M^{t}$ such that $M_{t}\simulates M^{t}_{t'}$ and $M^{t}_{t'}\models \phi$.\footnote{Note that in $M^t_{t'}$ the superscript $t$ refers to a state of $M$, but here serves only to identify a model $M^t$. The subscript $t'$, meanwhile, identifies a particular state of $M^t$, namely the one such that $M_{t}\simulates M^{t}_{t'}$.}

Every $M^{t}_{t'}$ satisfies some $\phi\in \Phi_a$, but there is no guarantee that every $\phi\in \Phi_a$ is satisfied in such a model. Let $\Phi_a^+\subseteq \Phi_a$ be the set of formulas that are so satisfied, and $\Phi_a^- = \Phi_a\setminus \Phi_a^+$. Because all $\phi\in \Phi_a$ are, by assumption, consistent, there exist models $M^\phi$ and $t_\phi$ such that $M^\phi_{t_\phi}\models \phi$.

Now, let $M'$ be the disjoint union of (i) the $M^t_{t'}$ for each $a\in A$ and $(s,t)\in R_a$ and (ii) the $M^{\phi}_{t_\phi'}$ for each $\phi\in \Phi_a^-$. Add a single state $s_0'$, with $V'(s_0')=V(s)$, and $a$-edges from $s_0'$ to every $M^t_{t'}$ where $(s,t)\in R_a$ and every $M^\phi_{t_\phi'}$ where $\phi\in \Phi_a^-$. Call the resulting model $M'=(S',R',V')$. Formally,
\begin{itemize}
	\item $S' = \{s_0'\} \cup \bigcup_{(s,t)\in R_a}\{(s',a,t)\mid s'\in S^t\}\cup \bigcup_{\phi\in \Phi_a^-}\{(s',a,\phi)\mid s'\in S^\phi\}$
	\item $R'_a = \{(s_0',(t',a,t))\mid (s,t)\in R_a\}\cup \{(s_0',(s_\phi',a,\phi))\mid \phi\in \Phi_a^-\}\cup$ 
	\\
	\hspace*{50pt}
	$\bigcup_{a\in A}\bigcup_{(s,t)\in R_a}\{((s_1',a,t),(s_2',a,t))\mid (s_1',s_2')\in R_a^t\}\cup$ 
	\\
	\hspace*{50pt}
	$\bigcup_{a\in A}\bigcup_{\phi\in\Phi_a^-}\{((s_1',a,\phi),(s_2',a,\phi))\mid (s_1',s_2')\in R_a^\phi\}$
	\item $V'(s') = \left\{ \begin{array}{ll}V(s_0) & \text{if }s'=s_0'\\ V^t(t')&\text{if } s' = (t',a,t)\\V^\phi(t')&\text{if } s' = (t',a,\phi)\end{array}\right.$
\end{itemize}
See Figure~\ref{fig:M'-construction2} for an example of this construction.

For every $(s_0',s')\in R_a'$, we have $s'=(t',a,t)$ or $s'=(s_\phi',a,\phi)$. In the first case, there is some $\phi\in \Phi_a^+$ such that $M'_{s'}\models \phi$, since $M'_{s'}$ is a copy of $M^t_{t'}$. In the second case, there is some $\phi\in \Phi_a^-$ such that $M'_{s'}\models\phi$, since $M'_{s'}$ is a copy of $M^\phi_{s_\phi'}$. So every $a$-successor of $s_0'$ satisfies some $\phi\in \Phi_a$.

Furthermore, for every $\phi\in \Phi_a^+$ there is some $t$ such that $(s_0,t)\in R_a$ and $M^t_{t'}\models \phi$, which implies that $M'_{(t',a,t)}\models \phi$ and therefore that $M'_{s_0'}\models \lozenge_a\phi$. For very $\phi\in \Phi_a^-$, we have $M^\phi_{s_\phi'}\models\phi$, which implies that $M'_{(s_\phi',a,\phi)}\models \phi$ and therefore that $M'_{s_0'}\models \lozenge_a\phi$. Taken together, this shows that every $\phi\in \Phi_a$ is satisfied in some $a$-successor of $s_0'$. We had already shown that every $a$-successor of $s_0'$ satisfies some $\phi\in \Phi_a$, so we have $M'_{s_0'}\models \Et_{a\in A}\cover_a\Phi_a$.

Now, we will show that $M_{s_0}\simulates M'_{s_0'}$. For each $M^t_{t'}$, we have $M_t\simulates M^t_{t'}$, let us call the witnessing simulation $Z^t$. Now, let $Z$ be the disjoint union of all $Z^t$ together with $(s_0,s_0')$. Formally, 
\[Z = \{(s,s_0')\}\cup \bigcup_{a\in A}\bigcup_{(s_0,t)\in R_a} \{(s,(s',a,t))\mid (s,s')\in Z^t\}.\] 
We claim that $Z$ is a simulation.

Take any $(x,x')\in Z$. We consider two cases. Firstly, suppose that $x'=s_0'$. Then $x=s_0$. We have $V(s_0)=V'(s_0')$ by construction, so in this case {\bf atoms} is satisfied. Secondly, suppose $x'\not = s_0'$. Then $x'$ must be of the form $(s',a,t)$, since $Z$ does not relate to any states of the form $(s',a,\phi)$. By the construction of $Z$, we then have $(x,s')\in Z$, which implies that $V(x)=V^t(s')$, and therefore also $V(x)=V'(x')$. So in this case {\bf atoms} is also satisfied. These two cases are exhaustive, so {\bf atoms} is satisfied for $Z$.

Then take any $(x_1,x_1')\in Z$ and $(x_1,x_2)\in R_a$. Again, we consider two cases. Firstly, suppose $x_1'=s_0'$. Then $x_1=s_0$, and $x_2=t$ for some $t$ such that $(s,t)\in R_a$. Then there is some $M^t_{t'}$ such that $(t,t')\in Z^t$. By the construction of $Z$ we have $(t,(t',a,t))\in Z$, and by the construction of $M'$ we have $(s_0',(t',a,t))\in R_a$. So {\bf forth} is satisfied in this case. Secondly, suppose $x_1'\not = s_0'$. Then $x_1'$ is of the form $x_1'=(s_1',b,t)$, where $(x_1,s_1')\in Z^t$. Because $Z^t$ is a simulation, there is some $s_2'$ such that $(s_1',s_2')\in R_a^t$ and $(x_2,s_2')\in Z^t$. We then also have $(s_1,(s_2',b,t))\in R_a'$ and $(x_2,(s_2',b,t))\in Z$, so {\bf forth} is satisfied in this case. The two cases are exhaustive, so {\bf forth} is satisfied for $Z$.

Because $Z$ satisfies {\bf atoms} and {\bf forth}, it is a simulation, so $M_s\simulates M'_{s_0'}$. Together with the previously shown $M'_{s_0'}\models \Et_{a\in A}\cover_a\Phi_a$, this implies that $M_{s_0}\models \smldia\Et_{a\in A}\cover_a\Phi_a$. This completes the proof that 
\[\models \Et_{a\in A}\square_a\bigvee_{\phi\in \Phi_a}\smldia\phi\rightarrow \Et_{a\in A}\cover_a\Phi_a,\] 
if all $\phi\in \Phi_a$ are consistent. We had already shown the other direction, so 
\[\models \smldia \Et_{a\in A}\cover_a\Phi_a \rightarrow \Et_{a\in A}\square_a\bigvee_{\phi\in \Phi_a}\smldia \phi,\]
when all $\phi$ are consistent. This shows that ${\bf SQ4_{cons}}$ is sound, which completes the soundness proof for {\bf SML}.
\end{proof}
\begin{figure}
\begin{center}
\begin{tikzpicture}[scale=0.9, transform shape]
\draw[rounded corners=1cm] (-1.5,1) -- (3.5,1) -- (3.5,-3.8) -- (-1.5,-3.8) -- cycle;
\node at (1,-4.1) {$M$};
\node[label=above:$s_0$] (s) at (1,0) {$\bullet$};

\node[label=below:$\begin{array}{c}t_1\\ \smldia\phi_1\end{array}$] (t1) at (-0.5,-2.2) {$\bullet$};
\node[label=below:$\begin{array}{c}t_2\\ \smldia\phi_1\end{array}$] (t2) at (1,-2.2) {$\bullet$};
\node[label=below:$\begin{array}{c}t_3\\ \smldia\phi_3\end{array}$] (t3) at (2.5,-2.2) {$\bullet$};

\draw[->] (s) -- (t1) node[midway,left] {$a$};
\draw[->] (s) -- (t2) node[midway,left] {$a$};
\draw[->] (s) -- (t3) node[midway,right] {$b$};

\draw[dashed] (4.1,1.2) -- (4.1,-5);

\draw[rounded corners=1cm] (4.8,1) -- (15.2,1) -- (15.2,-4.4) -- (4.8,-4.4) -- cycle;
\node at (10,-4.7) {$M'$};

\node[label=above:$s_0'$] (s0') at (10,0) {$\bullet$};

\node[label=below:$\begin{array}{c}t_1'\\\phi_1\end{array}$] (t1') at (6,-2.2) {$\bullet$};
\node[draw,ellipse,minimum width=1.5cm,minimum height=2.5cm,label=below:$M^{t_1}$] (M1) at (6,-2.5){};
\node[label=below:$\begin{array}{c}t_2'\\\phi_1\end{array}$] (t2') at (8,-2.2) {$\bullet$};
\node[draw,ellipse,minimum width=1.5cm,minimum height=2.5cm,label=below:$M^{t_2}$] (M2) at (8,-2.5){};
\node[label=below:$\begin{array}{c}s_{\phi_2}'\\\phi_2\end{array}$] (sphi2') at (10,-2.2) {$\bullet$};
\node[draw,ellipse,minimum width=1.5cm,minimum height=2.5cm,label=below:$M^{\phi_2}$] (M3) at (10,-2.5){};
\node[label=below:$\begin{array}{c}t_3'\\\phi_3\end{array}$] (t3') at (12,-2.2) {$\bullet$};
\node[draw,ellipse,minimum width=1.5cm,minimum height=2.5cm,label=below:$M^{t_3}$] (M4) at (12,-2.5){};
\node[label=below:$\begin{array}{c}s_{\phi_4}'\\\phi_4\end{array}$] (sphi4') at (14,-2.2) {$\bullet$};
\node[draw,ellipse,minimum width=1.5cm,minimum height=2.5cm,label=below:$M^{\phi_4}$] (M4) at (14,-2.5){};

\draw[->] (s0') -- (t1') node[midway,left] {$a$};
\draw[->] (s0') -- (t2') node[midway,left] {$a$};
\draw[->] (s0') -- (sphi2') node[midway,left] {$a$};
\draw[->] (s0') -- (t3') node[midway,right] {$b$};
\draw[->] (s0') -- (sphi4') node[midway,right] {$b$};

\draw[dotted] (s) edge[bend left] (s0');
\draw[dotted] (t1) edge[bend right] (t1');
\draw[dotted] (t2) edge[bend right] (t2');
\draw[dotted] (t3) edge[bend right] (t3');
\end{tikzpicture}
\end{center}

\caption{Schematic drawing of an example of $M'$ as in the proof of Proposition~\ref{prop:SMLsoundness}. Here, we take $A=\{a,b\}$, $\Phi_a=\{\phi_1,\phi_2\}$ and $\Phi_b=\{\phi_3,\phi_4\}$. For both $s_0$ and $s_0'$, all successors are drawn. In each successor of $s_0$, $\smldia\phi_i$ holds for some $i$. Note that the same $\smldia\phi_i$ may hold in multiple successors, such as $\smldia\phi_1$ in $t_1$ and $t_2$.
Each successor of $s_0'$ corresponds either to a successor of $s_0$ or to some $\phi_i\in \Phi_a^-\cup \Phi_b^-$. Dotted lines represent the relation $Z$. Note that every successor of $s_0$ is $Z$-related to some successor of $s_0'$ but not vice versa, so $Z$ is a simulation but not a bisimulation.}
\label{fig:M'-construction2}
\end{figure}
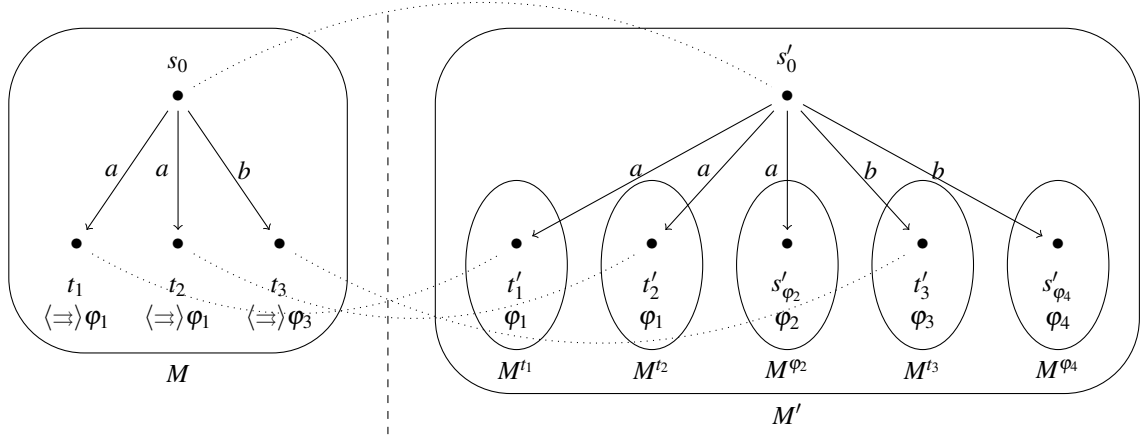

\begin{restatable}{proposition}{propROSMLsound}
\label{prop:ROSMLsoundness}
The axiomatization {\bf ROSML} is sound for $\lang$.
\end{restatable}
\begin{proof}
${\bf ROSML} = {\bf ML} + {\bf REF} + {\bf ORI} + {\bf SIM}$. Soundness of {\bf ML} is well known, and we showed soundness of {\bf REF}, {\bf ORI} and {\bf SQ1}--{\bf SQ3} previously. This leaves only the soundness of {\bf SQ4}.

For the left-to-right direction, suppose that $M_s\models \smldia\bigwedge_{a\in A}\cover_a\Phi_a$. Then each $\phi\in \Phi_a$ is consistent, so by the soundness of ${\bf SQ4_{cons}}$ we have $M_s\models \Et_{a\in A}\square_a\bigvee_{\phi\in \Phi_a}\smldia\phi$. Furthermore, because $\phi$ is consistent, there are $M^\phi$ and $s_\phi$ such that $M^\phi_{s_\phi}\models \phi$. 

Every pointed model is the refinement of the state of $M^\circ$ with the same valuation, so we have $M^\circ_{V(s_\phi)}\refines M^\phi_{s_{\phi}}$. It follows that $M^\circ_{V(s_\phi)}\models \rmldia \phi$. In turn, that implies that $M^\circ_{V(s)}\models \lozenge_a\rmldia\phi$, since every state of $M^\circ$ is an $a$-successor of every other state. This is true for every $\phi\in \Phi_a$, so $M^\circ_{V(s)}\models\Et_{\phi\in\Phi_a}\lozenge_a\rmldia\phi$, which implies that $M_s\models \originbox \Et_{\phi\in\Phi_a}\lozenge_a\rmldia\phi$. Finally, this is true for every $a\in A$, which together with the previous conclusion that  $M_s\models \Et_{a\in A}\square_a\bigvee_{\phi\in \Phi_a}\smldia\phi$ implies that
\[M_s\models \Et_{a\in A}(\square_a\bigvee_{\phi\in \Phi_a}\smldia\phi\wedge \originbox \Et_{\phi\in \Phi_a}\lozenge_a\rmldia\phi).\]
We have now shown that
\[\models \smldia\Et_{a\in A}\cover_a\Phi_a\rightarrow \Et_{a\in A}(\square_a\bigvee_{\phi\in \Phi_a}\smldia\phi\wedge \originbox \Et_{\phi\in \Phi_a}\lozenge_a\rmldia\phi).\]
For the right-to-left direction, suppose that $M_s\models \Et_{a\in A}(\square_a\bigvee_{\phi\in \Phi_a}\smldia\phi\wedge \originbox \Et_{\phi\in \Phi_a}\lozenge_a\rmldia\phi)$. Then $M_s\models \Et_{a\in A}\square_a\bigvee_{\phi\in \Phi_a}\smldia\phi$ and, for every $a\in A$, $M_s\models \originbox\Et_{\phi\in \Phi_a}\lozenge_a\rmldia\phi$. The latter implies that $\phi$ holds in some refinement of some state of $M^\circ$, so, in particular, $\phi$ is satisfiable.
From the soundness of ${\bf SQ4_{cons}}$, it then follows that $M_s\models \smldia\Et_{a\in A}\cover_a\Phi_a$. We have now shown that
\[\models \Et_{a\in A}(\square_a\bigvee_{\phi\in \Phi_a}\smldia\phi\wedge \originbox \Et_{\phi\in \Phi_a}\lozenge_a\rmldia\phi)\rightarrow \smldia\Et_{a\in A}\cover_a\Phi_a.\]
We already showed the other direction, so we have 
\[\models \smldia\Et_{a\in A}\cover_a\Phi_a\leftrightarrow \Et_{a\in A}(\square_a\bigvee_{\phi\in \Phi_a}\smldia\phi\wedge \originbox \Et_{\phi\in \Phi_a}\lozenge_a\rmldia\phi).\]
This shows that {\bf SQ4} is sound, which completes the soundness proof for {\bf ROSML}.
\end{proof}

\subsection{Completeness}


Having shown the soundness of the reduction axioms for \textbf{OML}, \textbf{RML}, \textbf{SML}, and \textbf{ROSML}, we can now establish the completeness of the corresponding axiomatizations. In all of the following completeness proofs, the general idea is to use the reduction axioms to translate a given formula of one of the new logics into an equivalent formula of \textbf{ML}, which is known to be complete. 


\begin{restatable}{proposition}{propRMLcomp}\label{prop:RMLComplete}
The axiomatization \textbf{RML} is complete for $\lang^{\Box\refines}$.
\end{restatable}

\begin{proof}
	Given $\rmldia \varphi$, the proof is by induction on the modal depth of $\varphi$. For the base case, $d(\varphi) = 0$, and hence $\rmldia \phi_0$ is equivalent to $\phi_0$ by \textbf{RQ1}. Now assume that we have $\rmldia \varphi$ with $d(\phi) = n$, and all $\rmldia \psi$ with $d(\psi) < n$ are equivalent to some formulas without refinement quantifiers. W.l.o.g. we may also assume that $\varphi$ is a combination of propositional atoms and formulas of the type $\Box_a \chi_1$ and $\Diamond_a \chi_2$. By \textbf{RE}, we can substitute $\phi$ with an equivalent formula in the disjunctive normal form:
	$$
\bigvee \left(
\begin{array}{l}
p_{11} \land ... \land p_{1m} \land \bigwedge_{a \in ag(\phi)}(\Diamond_a \chi_1 \land ... \land \Diamond_a \chi_l \land \Box_a \theta_1 \land ... \land \Box_a \theta_k) \\
...\\
p_{n1} \land ... \land p_{nm} \land \bigwedge_{a \in ag(\phi)}(\Diamond_a \chi_1 \land ... \land \Diamond_a \chi_l \land \Box_a \theta_1 \land ... \land \Box_a \theta_k)
\end{array}
\right),
$$   
where $ag(\phi)$ is the set of agents appearing in $\phi$.
By \textbf{RQ2}, we have that $\rmldia \phi$ is equivalent to 
	$$
\bigvee \left(
\begin{array}{l}
\rmldia (p_{11} \land ... \land p_{1m} \land \bigwedge_{a \in ag(\phi)}(\Diamond_a \chi_1 \land ... \land \Diamond_a \chi_l \land \Box_a \theta_1 \land ... \land \Box_a \theta_k)) \\
...\\
\rmldia(p_{n1} \land ... \land p_{nm} \land \bigwedge_{a \in ag(\phi)}(\Diamond_a \chi_1 \land ... \land \Diamond_a \chi_l \land \Box_a \theta_1 \land ... \land \Box_a \theta_k))
\end{array}
\right).
$$  
Observe that $\Diamond_a \chi_1 \land ... \land \Diamond_a \chi_l \land \Box_a \theta_1 \land ... \land \Box_a \theta_k$ is equivalent (via {\bf K}) to $\cover_a \Phi_a$ with $\Phi_a = \{\chi_1\land\bigwedge_{i=1}^k\theta_k, ..., \chi_l\land\bigwedge_{i=1}^k\theta_k, \bigwedge_{i=1}^k\theta_k\}$. 
This follows from the validities $(\Box_a\phi\land\Diamond_a\psi)\rightarrow \Diamond_a(\phi\land\psi)$:
$$
\begin{tabular}{lll}
1. & $\Box_a(\phi\rightarrow\lnot\psi)\rightarrow(\Box_a\phi \rightarrow\Box_a\lnot\psi)$ & \textbf{K}, \textbf{RE}\\
2. & $\Box_a(\lnot\phi\lor\lnot\psi)\rightarrow (\lnot\Box_a\phi\lor\lnot\Diamond_a\psi)$ & \textbf{K}, \textbf{RE}, \textbf{Prop}\\
3. & $(\Box_a\phi\land\Diamond_a\psi)\rightarrow \Diamond_a(\phi\land\psi)$ & \textbf{RE}, \textbf{Prop}
\end{tabular}
$$
as well as $(\Box_a\phi\land\Box_a\psi)\leftrightarrow\Box_a(\phi\land\psi)$, which is a well known validity of \textbf{ML}.

Then, by \textbf{RQ3}, we have that 
	$$
\bigvee \left(
\begin{array}{l}
p_{11} \land ... \land p_{1m} \land \rmldia\bigwedge_{a \in ag(\phi)}\cover_a \Phi_a\\
...\\
p_{n1} \land ... \land p_{nm} \land \rmldia\bigwedge_{a \in ag(\phi)}\cover_a \Phi_a
\end{array}
\right),
$$  
which, in turn, is equivalent, by \textbf{RQ4}, to 
	$$
\bigvee \left(
\begin{array}{l}
p_{11} \land ... \land p_{1m} \land \bigwedge_{a \in ag(\phi)}\bigwedge_{\chi\in \Phi_a}\Diamond_a \rmldia\chi\\
...\\
p_{n1} \land ... \land p_{nm} \land \bigwedge_{a \in ag(\phi)}\bigwedge_{\chi\in \Phi_a}\Diamond_a \rmldia\chi
\end{array}
\right).
$$  
The modal depth of all $\chi$ is at most $n-1$, and hence, by the induction hypothesis, for each $\rmldia \chi$ there is an equivalent quantifier-free formula.
	\end{proof}

With the similar approach as in the proof of Proposition \ref{prop:RMLComplete}, we can show the completeness of \textbf{SML}.

\begin{restatable}{proposition}{propSMLcomp}
\label{prop:SMLComplete}
The axiomatization \textbf{SML} is complete for $\lang^{\Box\simulates}$.
\end{restatable}
\begin{proof}
	The proof goes, \textit{mutatis mutandis}, exactly like the proof of Proposition \ref{prop:RMLComplete}. We use the induction on the modal depth of $\phi$ to show that  $\phi$, which may have simulation quantifiers, can be equivalently rewritten into a formula of modal logic. For the induction step $\smldia \phi$, we also substitute $\phi$ with an equivalent formula in disjunctive normal form, and then push $\smldia$ inside using $\textbf{SQ2}$, $\textbf{SQ3}$, and $\textbf{SQ4}_{\mathsf{cons}}$. As a consequence, formulas within the scope of $\smldia$ have lower modal depth, and the result follows from the induction hypothesis.
\end{proof}

The proof of completeness for \textbf{OML} is also a proof by elimination, in that we inductively show that every formula of $\lang^{\Box\circ}$ is provably (in \textbf{OML}) equivalent to a formula of $\lang^\Box$.

\begin{restatable}{proposition}{propOMLcomp}
\label{prop:OMLComplete}
The axiomatization \textbf{OML} is complete for $\lang^{\Box\circ}$.
\end{restatable}
\begin{proof}
The structure of the model $M^\circ$ has every relation $R^\circ_a$ as a total relation, so it is an equivalence class. 
This is reflected in the axiomatization $\textbf{OML}$ where we have the axioms \textbf{T} ($\Box_a\phi\rightarrow\phi$)
and \textbf{5} ($\Diamond_a\phi\rightarrow\Box_a\Diamond_a\phi$) inside the $\originbox$ operator as the axioms \textbf{OT} and \textbf{O5}.
Similarly, we have the necessitation rule and modus ponens rule in the scope of the $\originbox$ operator with the rules \textbf{OMP} and \textbf{ON}.
Therefore, for every formula $\phi$ that is valid for multi-agent S5 frames, there is a proof in \textbf{OML} of $\originbox\phi$.

We proceed to the elimination argument as follows. 
We work by induction, where the induction hypothesis is \textit{for all $\lang^\Box$ formulas $\phi$ with modal depth less than or equal to $n$, $\originbox{\phi}$ is provably equivalent to a formula of $\lang^{\Box}$}.
The base follows directly from \textbf{O1}. 
Now we assume that the induction hypothesis holds for $n$.

Applying the axioms \textbf{O1}, \textbf{ODual} and \textbf{ODisj} we can move the $\originbox$ operators up to the modalities.
That is
$$\originbox\phi\leftrightarrow\theta[x_i\backslash\originbox\Diamond_{a_i}\phi_i]_{i=1}^m$$
is a validity where $\theta\in\lang_0$ is a propositional formula.
Furthermore, noting that $\Diamond_a$ commutes with disjunctions 
($\Diamond_a(\phi\lor\psi)\leftrightarrow(\Diamond_a\phi\lor\Diamond_a\psi)$ is provable in \textbf{ML}), we can extend this result to say
$$\originbox\phi\leftrightarrow\theta[x_i\backslash\originbox\Diamond_{a_i}\bigwedge_{j=1}^k\chi^i_j]_{i=1}^m$$
where $\chi_j^i$ is either a literal (an atomic proposition or its negation), or a modal formula of one of the forms $\Box_b\xi$, or $\Diamond_b\xi$.

So, via some basic \textbf{ML} reasoning, it is now sufficient for us to show that formulas of the type 
\begin{equation}\label{eq:target}
\originbox\Diamond_a(\pi\land\bigwedge_{b\in A}\bigwedge_{i=1}^{\ell_b}\Diamond_b\chi^b_i\land\Box_b\bigwedge_{j=1}^{k_b}\xi^b_j)
\end{equation}
are equivalent to $\lang^\Box$ formulas, where $\pi$ is a conjunction of literals, and $\chi^b_i$ and $\xi^b_i$ are $\lang^\Box$ formulas of modal depth at most $n$.

By S5 reasoning, we have:
$$
\begin{array}{c}
\Diamond_a(\pi\land\bigwedge_{b\in A}\bigwedge_{i=1}^{\ell_b}\Diamond_b\chi^b_i\land\Box_b\bigwedge_{j=1}^{k_b}\xi^b_j)\\
\leftrightarrow\\
\Diamond_a(\pi\land\bigwedge_{b\in A\backslash\{a\}}\bigwedge_{i=1}^{\ell_b}\Diamond_b\chi^b_i\land\Box_b\bigwedge_{j=1}^{k_b}\xi^b_j)
\land
\bigwedge_{i=1}^{\ell_a}\Diamond_a\chi^a_i\land\Box_a\bigwedge_{j=1}^{k_a}\xi^a_j
\end{array}
$$
Therefore, by the above observation that all S5 validities are derivable in the scope of the $\originbox$ operator, and \textbf{ODual} with \textbf{ODisj}, we have:
$$
\begin{array}{c}
\originbox\Diamond_a(\pi\land\bigwedge_{b\in A}\bigwedge_{i=1}^\ell\Diamond_b\chi^b_i\land\Box_b\bigwedge_{j=1}^k\xi^b_j)\\
\leftrightarrow\\
\originbox\Diamond_a(\pi\land\bigwedge_{b\in A\backslash\{a\}}\bigwedge_{i=1}^\ell\Diamond_b\chi^b_i\land\Box_b\bigwedge_{j=1}^k\xi^b_j)
\land
\bigwedge_{i=1}^\ell\originbox\Diamond_a\chi^a_i\land\originbox\Box_a\bigwedge_{j=1}^k\xi^a_j
\end{array}
$$
As for all $i$, $\Diamond_a\chi^a_i$ and $\Box_a\xi^a_i$ have modal depth less than or equal to $n$, we can apply the induction 
hypothesis, so that they are provably equivalent to $\lang^\Box$ formulas. Therefore we have 
$$
\begin{array}{c}
\originbox\Diamond_a(\pi\land\bigwedge_{b\in A}\bigwedge_{i=1}^\ell\Diamond_b\chi^b_i\land\Box_b\bigwedge_{j=1}^k\xi^b_j)\\
\leftrightarrow\\
\originbox\Diamond_a(\pi\land\bigwedge_{b\in A\backslash\{a\}}\bigwedge_{i=1}^\ell\Diamond_b\chi^b_i\land\Box_b\bigwedge_{j=1}^k\xi^b_j)
\land \phi_a
\end{array}
$$
where $\phi_a$ is a $\lang^\Box$ formula. We can now apply \textbf{OExch}, noting that is is an equivalence, to derive:
\begin{equation}\label{eq:recurse}
\begin{array}{c}
\originbox\Diamond_a(\pi\land\bigwedge_{b\in A}\bigwedge_{i=1}^\ell\Diamond_b\chi^b_i\land\Box_b\bigwedge_{j=1}^k\xi^b_j)\\
\leftrightarrow\\
\originbox\Diamond_b(\pi\land\bigwedge_{b\in A\backslash\{a\}}\bigwedge_{i=1}^\ell\Diamond_b\chi^b_i\land\Box_b\bigwedge_{j=1}^k\xi^b_j)
\land \phi_a
\end{array}
\end{equation}

Noting the similarity of the first line of (\ref{eq:recurse}) and the first conjunct of the second line of (\ref{eq:recurse}), where $b\in A\backslash\{a\}$, we can recursively apply this argument for each modality $a\in A$ to derive 
$$
\originbox\Diamond_a(\pi\land\bigwedge_{b\in A}\bigwedge_{i=1}^\ell\Diamond_b\chi^b_i\land\Box_b\bigwedge_{j=1}^k\xi^b_j)
\leftrightarrow
\originbox\Diamond_a \pi\land\bigwedge_{a\in A}\phi_a
$$
Here, the $a$ in $\originbox\Diamond_a\pi$ is arbitrary. It will be the last modality left in the set $A$, but via the \textbf{OExch} axiom it can be changed to any modality.
As $\pi$ is a conjunction of literals, either it contains an atom and its negation, in which case we have:
$$
\originbox\Diamond_a(\pi\land\bigwedge_{b\in A}\bigwedge_{i=1}^\ell\Diamond_b\chi^b_i\land\Box_b\bigwedge_{j=1}^k\xi^b_j)
\leftrightarrow \bot
$$
or it matches the condition for \textbf{OFull}, in which case we can derive $\originbox\Diamond_a\pi$ so
$$
\originbox\Diamond_a(\pi\land\bigwedge_{b\in A}\bigwedge_{i=1}^\ell\Diamond_b\chi^b_i\land\Box_b\bigwedge_{j=1}^k\xi^b_j)
\leftrightarrow \bigwedge_{a\in A}\phi_a
$$
This is sufficient to show $\lang^{\Box\circ}$ formulas of the kind in equation (\ref{eq:target}) are equivalent to $\lang^\Box$ formulas
which concludes the completeness proof.
\end{proof}

Finally, having established the way to eliminate refinement and simulation quantifiers, as well as the origin modality, in proofs of Propositions \ref{prop:RMLComplete}, \ref{prop:SMLComplete}, and \ref{prop:OMLComplete}, the completeness of \textbf{ROSML} follows.

\begin{restatable}{proposition}{propROSMLcomp}
\label{prop:ROSMLComplete}
The axiomatization \textbf{ROSML} is complete for $\lang$.
\end{restatable}

\begin{proof}
The proof is similar to those of Propositions \ref{prop:RMLComplete}, \ref{prop:OMLComplete}, and \ref{prop:SMLComplete}. Yet again, we recursively eliminate quantifiers and the origin modality from the inside-out, starting from the innermost occurrence of the operator in question. Hence, for the induction case $\rmldia \phi$ we proceed exactly as in the proof of Proposition \ref{prop:RMLComplete}. For the case $\smldia \phi$, we follow the strategy in the proof of Proposition \ref{prop:SMLComplete}. Finally, for the origin modality we follow the proof of Proposition \ref{prop:OMLComplete}.
\end{proof}

\section{Conclusion and Discussion}

We have presented several novel logics for reasoning about information growth and information loss. Refinement modal logic RML quantifies over information growth. We gave a novel axiomatization for that logic, slightly different from \cite{bozzellietal.inf:2014}. Simulation modal logic SML quantifies over information loss. We presented two axiomatizations for (or containing) SML, \textbf{SML} and \textbf{ROSML}. The first, \text{SML}, has consistency requirements on formulas in one of the axioms and is a variation of \cite{xingetal:2019}. The second, \textbf{ROSML}, does not have that consistency requirement but apart from simulation quantifiers also has refinement quantifiers and the novel origin modality. The origin modality checks whether the formula bound by it was true in the model before any of the agents received any factual information: the mutual factual ignorance model. We also gave an axiomatization \textbf{OML} for origin modal logic OML with only this origin modality but without refinement and simulation quantifiers.   

There is ample scope for further research.
The semantics of the refinement and simulation quantifiers depend on the class of models it quantifies over. In our work, we have considered the class K of all models. 
Hence, an immediate research problem is the axiomatization of \textit{simulation epistemic logic}, i.e.\ SML with simulation modalities quantifying over S5 models, that is, epistemic models with equivalence accessibility relations. The axiomatization of a corresponding \textit{refinement epistemic logic} where refinement modalities quantify over S5 models has been presented in \cite{hales.aiml:2012,halesetal:2011, hales:2016}; these works also contain refinement modal logics over other classes of models, such as KD45 encoding consistent belief. 

Finally, it would be of clear interest to determine the complexity of satisfiability of SML. Let us survey such results for RML. The satisfiability problem of single-agent RML is AEXP$_{pol}$ 
and conjectured to be non-elementary for multi-agent RML \cite{bozzellietal.tcs:2014}. It is slightly unclear if the latter also holds for the RML of our contribution with refinement quantifiers for the set of all agents only. In \cite{achilleosetal:2013} a PSPACE upper bound is given for complexity of the existential fragment of single-agent RML (a negation normal form wherein only diamonds $\rmldia$ may occur). Complexities have also been reported for the logic $\mu$-RML extending RML with fixpoints \cite{bozzellietal.inf:2014}, and in work in progress building on the fixpoint extension \cite{xing22} of covariant-contravariant refinement modal logic \cite{xingetal:2019} (with modalities combining refinement and simulation). 

\paragraph*{Acknowledgements} We thank the TARK reviewers for their comments. Although based in the past of \cite{bozzellietal.inf:2014}, this research really got off the ground during Tim's visit to Hans in Toulouse, in April 2024, followed by subsequent one-on-one meetings of two collaborators, and also involving Rustam and Louwe, and both in Bergen and Prague, while the four of us were never together in the same place and time. Hans van Ditmarsch acknowledges the support of a grant {\em CIMI Mobilit\'e Internationale} from Universit\'e Paul Sabatier to visit Bergen, and support from the {\em Wolfgang Pauli Institute} (WPI) in Vienna where he was a Fellow during the completion of this research.

\bibliographystyle{eptcs}
\bibliography{biblio2024TARK}
\end{document}